%% file: arxiv.tex
\documentclass[sigconf,nonacm]{acmart}

\AtBeginDocument{%
  }

\setcopyright{acmlicensed}
\copyrightyear{2025}
\acmYear{2025}
\acmDOI{XXXXXXX.XXXXXXX}
\acmConference[PPDP'25]{The 27th International Symposium on Principles and
Practice of Declarative Programming}{September 10--11, 2025}{Rende, Italy}
\acmISBN{978-1-4503-XXXX-X/2018/06}

\usepackage{proof}
\usepackage{mathtools}
\usepackage{xfrac}
\usepackage{stmaryrd}
\usepackage[all,cmtip]{xy}
\usepackage{listings}
\input{macros}
\newcommand{\ssto}[1][]{~\to^{#1}~}
\newcommand{\bsto}{~\Downarrow~}

\newcommand{\skp}{\mathit{skip}}
\newcommand{\err}{\mathit{err}}
\newcommand{\sep}{\kern1pt , \kern-1pt}
\newcommand{\pmap}{ \xrightharpoonup{\hspace{0.1cm}} }
\newcommand{\nline}{\vspace{-3mm}}
\newcommand{\prem}[1]{(#1)}

\newcommand{\lrule}[3]{\textbf{#1}\quad\frac{#2}{#3}}

\begin{document}

\title{An Adequate While-Language for Stochastic Hybrid Computation}

\author{Renato Neves}
\email{nevrenato@di.uminho.pt}
\affiliation{%
        \institution{University of Minho \& INESC-TEC}
  \city{Braga}
  \country{Portugal}
}

\author{José Proença}
\email{jose.proenca@fc.up.pt}
\affiliation{%
  \institution{University of Porto \& CISTER}
  \city{Porto}
  \country{Portugal}}

\author{Juliana Souza}
\email{juliana.p.souza@inesctec.pt}
\affiliation{%
  \institution{University of Minho \& INESC-TEC}
  \city{Braga}
  \country{Portugal}
}

\renewcommand{\shortauthors}{ Neves et al. }

\begin{abstract}
       We introduce a language for formally reasoning about programs that
       combine differential constructs with probabilistic ones. The language
       harbours, for example, such systems as adaptive cruise controllers,
       continuous-time random walks, and physical processes involving multiple
       collisions, like in Einstein's Brownian motion.

       We furnish the language with an operational semantics and use it to
       implement a corresponding interpreter. We also present a complementary,
       denotational semantics and establish an adequacy theorem between both
       cases. 
\end{abstract}


\maketitle

\section{Introduction}

\noindent 
\textbf{Motivation.}
This paper aims at combining two lines of research in programming theory --
hybrid and probabilistic programming. Both paradigms are rapidly
proliferating and have numerous applications (see \eg\
\cite{platzer11,heunen17,dahlqvist19,goncharov20}), however, despite increasing
demand, their combination from a programming-oriented perspective is rarely
considered.

Examples abound on why such a combination is crucial, and even extremely simple
cases attest to this. Let us briefly analyse one such case. The essence of
hybrid programming is the idea of mixing differential constructs with classical
program operations, as a way of modelling and analysing computational devices
that closely interact with physical processes, such as movement, energy, and
electromagnetism. One of the simplest tasks in this context is to move a
stationary particle (for illustrative purposes one can regard it as a vehicle)
from position zero to position, say, three. A simple hybrid program that
implements this task is the following one.
\begin{lstlisting}[  ]
p := 0 ;
v := 0 ;
p' = v, v' = 1  for x  ;
p' = v, v' = -1 for y  
\end{lstlisting}
Observe the use of variable assignments which set the vehicle's position
(\lstinline|p|) and velocity (\lstinline|v|) to zero, and note as well the
presence of the sequencing operator (\lstinline|;|). Most notably, the last two
instructions are the aforementioned differential constructs which in this case
describe the continuous dynamics of the vehicle at certain stages of the
program's execution. Specifically \lstinline|p' = v, v' = 1 for x| states that
the vehicle will \emph{accelerate at the rate of $1\sfrac{m}{s^2}$ for
\lstinline|x| seconds} while the last instruction states that it will
\emph{deccelerate at the rate of $-1\sfrac{m}{s^2}$ for \lstinline|y| seconds}.
The goal then is to `solve' the program for \lstinline|x| and \lstinline|y| so
that the vehicle moves and subsequently stops precisely at position three.

Now, the reader has probably noticed that the program just described is an
idealised version of reality: there will be noise in the vehicle's actuators,
which will cause fluctuations in the acceleration rates, and the switching time
between one continuous dynamics to the other is not expected to be precisely
\lstinline|x| seconds but a value close to it. In face of this issue it is
natural to introduce uncertainty factors in the previous program and change
the nature of the question \emph{``will my vehicle be at position three at
\lstinline|x|+\lstinline|y| seconds?"} to a more probabilistic one, where we
ask about probabilities of reaching the desired position instead.

Remarkably, the alternative approach of simply considering
\emph{discretisations} of hybrid programs combined with probabilistic
constructs (\ie\ using \emph{purely} probabilistic programming) does not
work in general -- choosing suitables sizes for the discrete steps can be as
hard or even harder than taking the continuous variant, particularly when
so-called Zeno behaviour is present~\cite{platzer18}.

\noindent
\textbf{Contributions.}
We contribute towards a programming theory of stochastic hybrid computation --
\ie\ the combination of hybrid with probabilistic programming.  Following
traditions of programming theory, we first introduce a simple while-language on
which to study stochastic hybrid computation. Our language extends the original
while-language~\cite{winskel93,reynolds98} merely with the kind of differential
construct just seen and with random sampling~\cite{kozen79}.  Its simplicity is
intended, so that one can focus on the core essence of stochastic hybrid
computation, but we will see that despite such it already covers a myriad of
interesting and well-known examples.

We then furnish the language with an  operational semantics, so that we can
formally reason about stochastic hybrid programs. Among other things, we use
the semantics to extend \emph{Lince} -- an existing interpreter of hybrid
programs~\cite{goncharov20,mendes24} -- to an interpreter of stochastic hybrid
ones. We will show how this interpreter can be used to automatically present
statistical information about the program under analysis. All examples in this
paper were animated using our extended Lince, available
online.\footnote{\url{http://arcatools.org/lince} -- plots were adapted to avoid relying on colours}

Finally, following the basic motto in programming theory that different
semantic approaches are necessary to fully understand the computational
paradigm at hand, we introduce a compositional, denotational semantics for our
language. In a nutshell, for a given initial state $\sigma$ a program
denotation $\sem{\text{\lstinline|p|}}$ will correspond to a \emph{Markov
process}~\cite{panangaden09} -- intuitively its outputs are given in the form
of a probability distribution and are time-dependent. The semantics is built in
a systematic way, using basic principles of monad
theory~\cite{moggi:1989,girymonad}, measure theory, and functional
analysis~\cite{dudley02,aliprantis06,panangaden09}, to which we can then recur
(via the semantics) to derive different results about stochastic hybrid
computation. We end our contributions with the proof of an adequacy theorem
between the operational semantics and the denotational counterpart.

\noindent
\textbf{Related work.}
Whilst research on probabilistic programming is extensive (see \eg\
\cite{kozen79,culpepper17,barthe20,heunen17,dahlqvist19}), work on marrying it
with hybrid programming is much scarcer and mostly focussed on (deductive)
verification. The two core examples in this line of research are~\cite{peng15}
and~\cite{platzer11}. The former presents an extension of the process algebra
\textsc{CSP} that harbours both probabilistic and differential constructs.
Among other things it furnishes this extension with a process-algebra like,
transition-based semantics -- which although quite interesting for verification
purposes is less amenable to operational perspectives involving \eg\
implementations. It presents moreover a corresponding proof system for
reasoning about certain kinds of Hoare triple.  The latter \emph{op.
cit.}~\cite{platzer11} extends the well-known \emph{differential dynamic logic}
with probabilistic constructs.  This logic is based on a Kleene-algebraic
approach which while has resulted in remarkable progress w.r.t. verification,
it is also known to have fundamental limitations in the context of hybrid
programming, particularly in the presence of non-terminating behaviour which is
frequent in this domain (see details for example
in~\cite{hofner11,hofner_phdthesis,goncharov20}).

A computational paradigm related to hybrid programming is reactive
programming~\cite{bainomugisha13,pembeci03}. The latter is strongly oriented to
the notions of signal and signal transformer (thus in the same spirit of
Simulink~\cite{zou15}) and puts emphasis on a declarative programming style (as
opposed to a more imperative one, like ours). It has been proposed as a way of
modelling hybrid systems, among other things, and recently was extended to a
probabilistic setting although no differentiation was
involved~\cite{baudart20}.

\noindent
\textbf{Document structure.} 
Section~\ref{sec:lng} introduces our stochastic language, its operational
semantics, and corresponding interpreter.  Section~\ref{sec:bck} recalls a
series of measure-theoretic foundations for developing our denotational
semantics -- which is then presented in Section~\ref{sec:denot} together with
the aforementioned adequacy theorem.  Section~\ref{sec:conc} discusses future
work and concludes.

We assume from the reader knowledge of probability theory~\cite{dudley02},
monads~\cite{moggi:1989,girymonad,hofmann14}, and category
theory~\cite{maclane98}. A monad will often be denoted in the form of a Kleisli
triple, \ie\ $(T,\eta^T,(-)^{\klcomp^T})$, and whenever no ambiguities arise we
will omit the superscripts in the unit and Kleisli operations. Similarly we
will sometimes denote a monad just by its functorial component $T$. Still about
notation, we will denote respectively by $\inl$ and $\inr$ the left $X \to X +
Y$ and right $Y \to X + Y$ injections into a coproduct. We will denote
measurable spaces by the letters $X$, $Y$, $Z$\dots and whenever we need to
explicitly work with the underlying $\sigma$-algebras we will use instead
$(X,\Sigma_X)$, $(Y,\Sigma_Y)$, $(Z,\Sigma_Z)$, and so forth.

\section{The language and its operational semantics}
\label{sec:lng}
We now introduce our stochastic hybrid language.  In a nutshell, it extends the
hybrid language in~\cite{goncharov20,mendes24} with an instruction for sampling
from the uniform continuous distribution over the unit interval $[0,1]$. 

We start by postulating a finite set
$\{$\lstinline[mathescape]|x$_1$,...,x$_n$|$\}$
of variables and a stock of \emph{partial} functions 
\lstinline|f|
$: \Reals^n \pmap \Reals$ that contains the usual arithmetic operations,
trigonometric ones, and so forth. As usual partiality will be crucial for
handling  division by $0$, logarithms, and square roots, among other things. We
then define expressions and Boolean conditions via the following standard BNF
grammars,
\begin{align*}
        & \text{\lstinline|e|}
        ::= \text {\lstinline|x|}
        \mid
        \text{\lstinline|f(e,...,e)| }
        \\
        & \text{\lstinline|b|}
        ::= 
        \text{\lstinline|e <= e| }
        \mid 
        \text{ \lstinline|b && b| }
        \mid
        \text{ \lstinline!b || b! }
        \mid
        \text{ \lstinline |tt| }
        \mid
        \text{ \lstinline |ff|}
\end{align*}
where \lstinline|x| is a variable in the stock of variables and \lstinline|f|
is a function in the stock of partial functions. Programs are then built
according to the two BNF grammars below.
\begin{align*}
        & 
        \text{\lstinline|a|} 
        ::= 
        \text{\lstinline[mathescape]|x$_1$' $\> \>$= e ,..., x$_n$' $\> \>$ = e for e|}
        \mid \text{\lstinline|x := e |}
        \mid 
        \text{\lstinline|x := unif(0,1)|}
        \\
        & \text{\lstinline|p|} 
        ::= 
        \text{\lstinline|a|} 
        \mid 
        \text{\lstinline|p ; p|}
        \mid 
        \text{\lstinline|if b then p else p|}
        \mid 
        \text{\lstinline[mathescape]|while b do $\> \> \{ \>$ p $\> \}$ |}
\end{align*}
The first grammar formally describes the three forms that an atomic program 
\lstinline|a| 
can take: \emph{viz.} a differential operation -- expressing a system's
continuous dynamics -- that will `run' for \lstinline|e| time units, an
assignment, and the aforementioned sampling operation. The second grammar
formally describes the usual program constructs of imperative
programming~\cite{winskel93,reynolds98}.

We proceed by introducing some syntactic sugar relative to the differential
operations and sampling. We will use this sugaring later on to provide further
intuitions about the language.  

First, observe that the language supports \emph{wait
calls} by virtue of the instruction
\lstinline[mathescape]|x$_1$' = 0 ,..., x$_n$' ${\> \> }$= 0 for e|. 
The latter states that the system is halted for
\lstinline|e|
time units, as no variable can change during this time period. The operation
will be denoted by the more suggestive notation
\mbox{\lstinline|wait e|}.
Next, although we have introduced merely sampling from the uniform distribution
over $[0,1]$, it is well-known that other kinds of sampling can be encoded from
it. For example, given two real numbers
\lstinline|a| $\leq$ \lstinline|b|
one can effectively sample from the uniform distribution over
$[$\lstinline|a|$,$ \lstinline|b|$]$
via the sequence of instructions,
\begin{lstlisting}
x := unif(0,1) ; x := (b - a) * x + a
\end{lstlisting}
For simplicity we abbreviate such sequence to \lstinline|x := unif(a,b)|. In
the same spirit, it will be useful to sample from the exponential distribution
with a given rate 
\lstinline|lambda| $>0 $,
in which case we write, 
\begin{lstlisting}[mathescape]
x := unif(0,1) ; x := - ln(x)/lambda
\end{lstlisting}
and abbreviate this program to
\lstinline|x := exp(lambda)|.
Next, in order to sample from normal distributions we resort for example to the
Box-Muller method~\cite{devroye86}. Specifically we write, 
\begin{lstlisting}[mathescape]
x1 := unif(0,1) ;
x2 := unif(0,1) ;
x  := sqrt(-2 * (ln x1)) * cos(2 * pi * x2)
\end{lstlisting}
and suggestively abbreviate the program to \lstinline|x := normal(0,1)|. The
latter amounts to sampling from the normal distribution with mean $0$ and
standard deviation $1$. Note that sampling from a normal distribution with mean
\lstinline|m| and standard deviation \lstinline|s| is then given by, 
\begin{lstlisting}[mathescape]
x := normal(0,1) ; x := m + s * x 
\end{lstlisting}
We abbreviate this last program to \lstinline|x := normal(m,s)|.  We encode
Bernoulli trials in our language standardly.  Specifically Bernoulli trials,
denoted by
\lstinline|bernoulli(r,p,q)|,
with \lstinline|r| $\in [0,1]$ and \lstinline|p|, \lstinline|q| two programs,
are encoded as,
\begin{lstlisting}[mathescape]
x := unif(0,1) ; if x <= r then p else q
\end{lstlisting}
which denotes the evaluation of \lstinline|p| with
probability \lstinline|r| and \lstinline|q| with probability 
\lstinline|1 - r|. 
Finally, we will also resort to the usual syntactic sugar constructs in
imperative programming, \eg\
\lstinline|x := x + 1|
as \lstinline|x++| and so forth.

We are ready to introduce a series of examples written in our programming
language.  In order to render their descriptions more lively we complement the
latter with analysis information given by the aforementioned interpreter.  The
interpreter as well as the examples are available online at
\url{http://arcatools.org/lince}.

\newcommand{\newPlot}[3][height=45mm]{
  \begin{figure}[h!]
  \centering\includegraphics[#1]{#2.pdf}
  \caption{#3}
  \label{fig:ex-#2}
  \end{figure}
}
\begin{example}[Approximations of normal distributions via random walks]
        We start with a very simple case that does not involve any
        differential operation. Specifically we approximate the standard
        normal distribution via a random walk -- a very common procedure in
        probabilistic programming~\cite{klafter11,barthe20}. Note that this is
        different from the previous sampling operation
        \lstinline|x := normal(0,1)|,
        in that it does not involve any trigonometric or logarithmic operation;
        furthermore the resulting distribution will always be discrete.
        The general idea is to write down the program below.
\begin{lstlisting}[mathescape]
x := 0 ; c := 0 ;
while c <= n do {
      bernoulli(1/2, x++, x--) ; c++
} ;
x := x/sqrt(n)
\end{lstlisting}
Via an appeal to the central limit theorem~\cite{dudley02} one easily sees that
the program approximates the expected normal distribution. The parameter
\lstinline|n| refers to the degree of precision, getting a perfect result as
\lstinline|n| $\to \infty$. 

\end{example}
\newPlot[width=80mm]{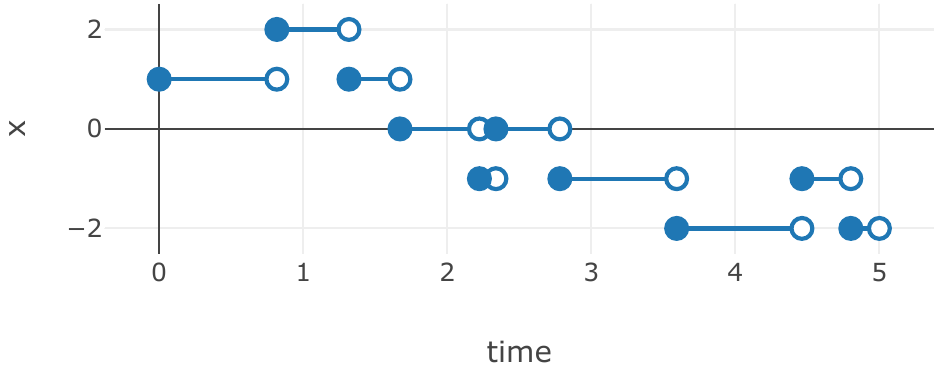}{An execution sample of a continuous-time random walk in which the
waiting time is given by sampling from the uniform distribution on $[0,1]$.}
\begin{example}[Continuous-time random walks]
        \label{ctrw}
We now shift our focus from random walks to continuous-time
ones~\cite{klafter11} -- a natural generalisation that introduces uncertainty
in the number of steps a walker performs in a given time interval.  These are
particularly useful for studying anomalous diffusion patterns (\ie\ the mean
squared displacement), with applications ranging from biology to
finance~\cite{klafter11}. A very simple example in our language is as follows.
\begin{lstlisting}[  ]
x := 0 ; 
while tt {
      bernoulli(1/2, x++, x--) ;
      d := unif(0,1) ;
      wait d
}
\end{lstlisting}
The prominent feature is that the walker now waits -- according to the uniform
distribution on $[0,1]$ -- before jumping. A helpful intuition from Nature is
to think for example of a pollinating insect that jumps from one flower to
another.  

Whilst we do not aim at fully exploring the example here, a quick inspection
tells that the average waiting time will be $\sfrac{1}{2}$ and thus the
diffusion pattern of this stochastic process grows linearly in time.  In
accordance to this, Figure~\ref{fig:ex-2.2-bw} presents an execution sample w.r.t.
the first 5 time units (horizontal axis) of the process: the plot was given by
our interpreter, and indeed shows 10 jumps performed during this period.

\end{example} 

\newPlot[width=80mm]{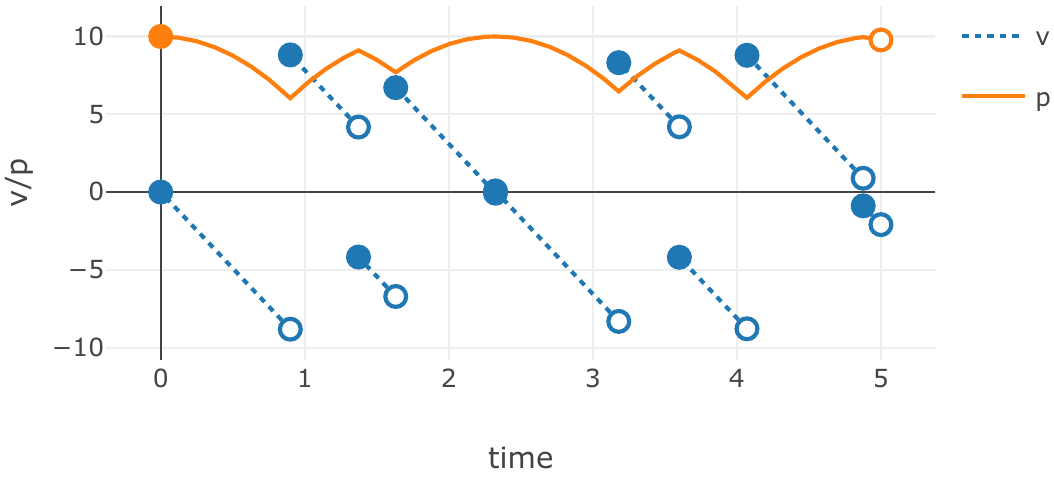}{An execution sample of the ball's position (\textsf{p})
and velocity (\textsf{v}) during the first 5 time units.}

\begin{example}[A ball and random kicks]
        \label{ball}
        Let us now consider one of the classical examples in hybrid systems
        theory, \emph{viz.} a bouncing ball~\cite{neves18,platzer18}. As usual
        we express the ball's continuous dynamics via a system of differential
        equations (those of motion) and jumps as assignments that `reset'
        velocity. In fact we will consider a variant in which there is no
        ground for the ball to bounce off. Instead it will be kicked
        randomly, as encoded in the following program.
        \begin{lstlisting}
p := 10 ; v := 0 ;
while tt do {
      d := unif(0,1) ;
      p' = v, v' = -9.8 for d ;
      v := -v 
}
        \end{lstlisting}
        In a nutshell, the ball moves according to the system of differential
        equations until it is kicked up (or down) as dictated by
        \lstinline|v := -v|.
        The duration between jumps is random, again with a mean time of
        $\sfrac{1}{2}$. Figure~\ref{fig:ex-2.3-bw} presents an execution sample
        of the bouncing ball during the first 5 time units.
\end{example}

\newPlot[height=55mm]{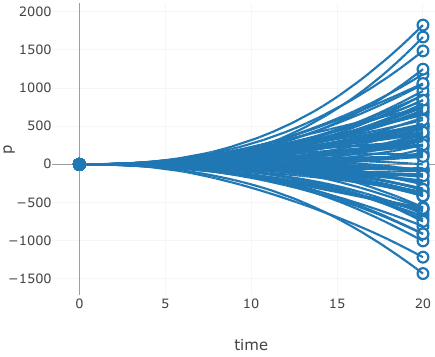}{Multiple execution samples of the particle's position
      overlayed, in order to depict how the position's probability mass 
      spreads over space w.r.t time.}

\begin{example}[Einstein's Brownian motion]
        \label{Einstein}
        Since the last example already moved us so close to it, we might as
        well consider Einstein's thesis for the cause of Brownian
        motion~\cite{einstein05}.  Essentially he posited that the (apparent)
        erratic motion of a particle suspended in a fluid is due to invisible
        collisions with atoms and molecules in the liquid.  In the
        one-dimensional setting, one can describe a particular instance of this
        stochastic process as follows.
        \begin{lstlisting}
p := 0 ; v := 0 ; a := 0 ;
while tt do {
      d := exp(lambda) ;
      bernoulli (1/2, a--, a++) ;
      p' = v, v' = a for d 
}
        \end{lstlisting}
        Note that random collisions are here encoded in the form a Bernoulli
        trial, and that their frequency is given by the Poisson process prescribed by
        the sampling operation.  Each collision causes a bump in the
        acceleration (which will either be incremented or decremented).
        Figure~\ref{fig:ex-2.4-100runs} then presents multiple execution samples
        overlayed on top of each other, in order to depict how the probability
        mass of the particle's position spreads over space w.r.t. time. 
\end{example}

\newPlot[width=80mm]{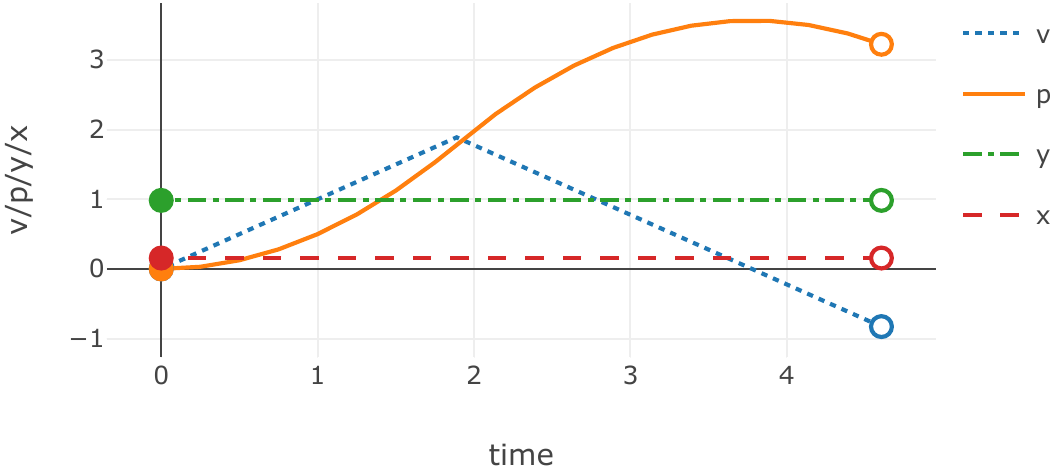}{Execution sample of the particle's position (\textsf{p})
and velocity (\textsf{v}).}

\begin{example}[Positioning of a particle]
        \label{ex:car}
        We now revisit the task of moving a stationary particle from position
        $0m$ to position $3m$, using acceleration rates $a = 1m/s^2$ and $a =
        -1m/s^2$. Recall that the respective program consists in accelerating
        (with rate $a$ $m/s^2$) and then decelerating ($-a$ $m/s^2$) the
        particle for prescribed durations \lstinline|x| and \lstinline|y|. Now,
        since $a$ and $-a$ have the same magnitude we can safely assume that
        \lstinline|x| $=$ \lstinline|y|. Such durations are then analytically deduced
        via the observation that,
        \[
        dist = \int^t_0 v_a(x) \, d x \> + \> \int^t_0 v_{-a}(x) \, dx 
        \]
        where $v_a(x) = a \cdot x$ and $v_{-a}(x) = v_a(t) + -a \cdot x$ are
        respectively the particle's velocity functions w.r.t. the time
        intervals $[0,t)$ and $[t,2\cdot t]$. In this context, $t$ is the value
        (\ie\ the duration \lstinline|x|) that we wish to find out (see further
        details in~\cite{mendes24}). Then observe that the value $dist$
        corresponds to the area of a triangle with basis $2 \cdot t$ and height
        $v_a(t)$.  This geometric shape yields the equations,
        \[
        \begin{cases}
                dist & = \frac{1}{2} \cdot (2 \cdot t) \cdot v_a(t) 
                \hspace{0.2cm} \text{\textbf{(area)}} \\
                v_a(t) & = a \cdot t \hspace{0.2cm} \text{\textbf{(height)}} 
        \end{cases}
        \hspace{0.5cm} \Longrightarrow \hspace{0.5cm}
        t = \sqrt{ \frac{dist}{a} }
        \]
        Finally note that for $dist = 3$ and $a = 1$ we obtain $t= \sqrt{3}$.
        Thus, ideally we would like to set the durations of both acceleration
        and deceleration to $\sqrt{3}$. This would then give rise to a total
        duration of $2\sqrt{3}$ and the particle would stop precisely at
        position $3$. However in reality we should expect a small error
        in the durations of such instructions. To this effect, in the program
        below we add an uncertainty factor to the calculated durations
        \lstinline|x| and \lstinline|y|.
\begin{lstlisting}[  ]
x := exp(2) + sqrt(3); 
y := exp(2) + sqrt(3); 
p' = v, v' = 1  for x ;
p' = v, v' = -1 for y
\end{lstlisting}
Figure~\ref{fig:ex-2.5-bw} presents an execution sample where we see the effects
of the small errors in the durations: specifically the program at some point
exceeds position 3 and terminates slightly above it. Note also that the 
expected shape of the velocity function is lost.

\end{example}

\begin{figure}[h!]
  \centering
  \includegraphics[width=80mm]{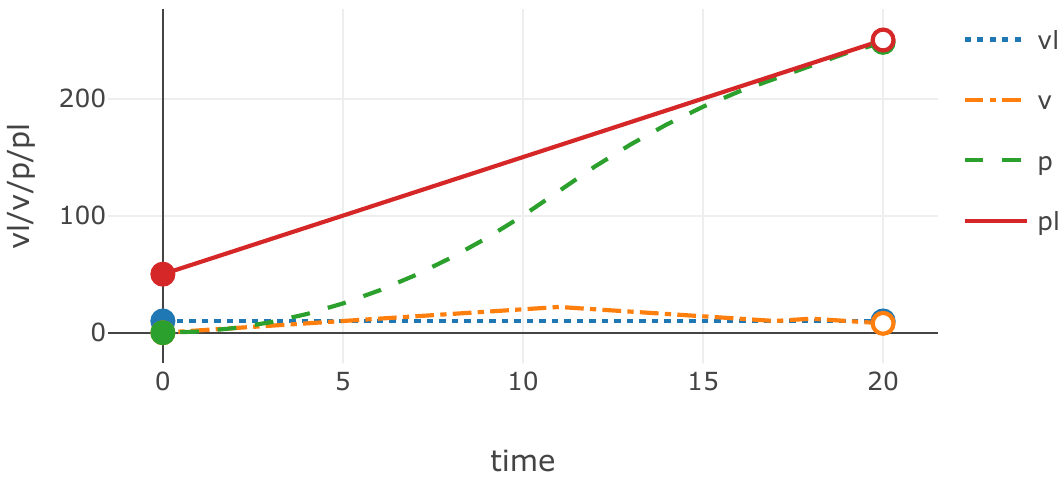}\\
  \includegraphics[width=80mm]{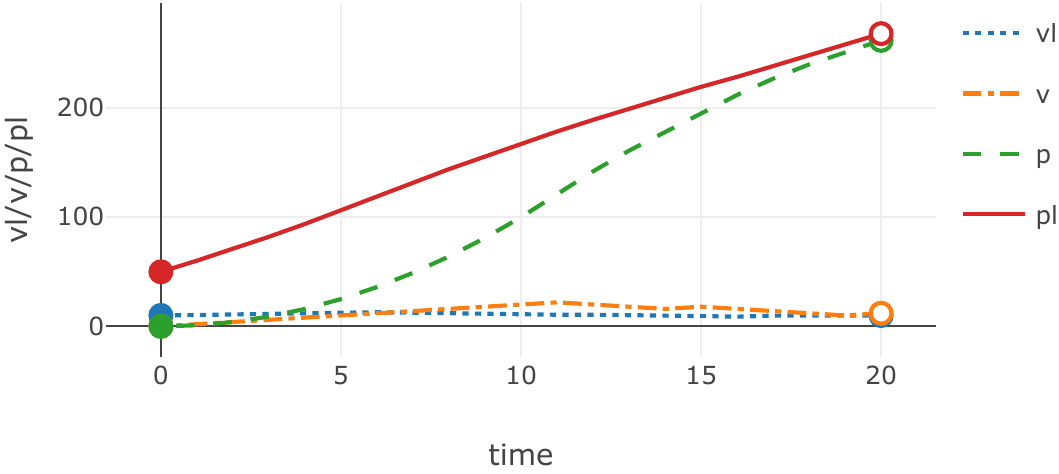}\\
  \includegraphics[width=80mm]{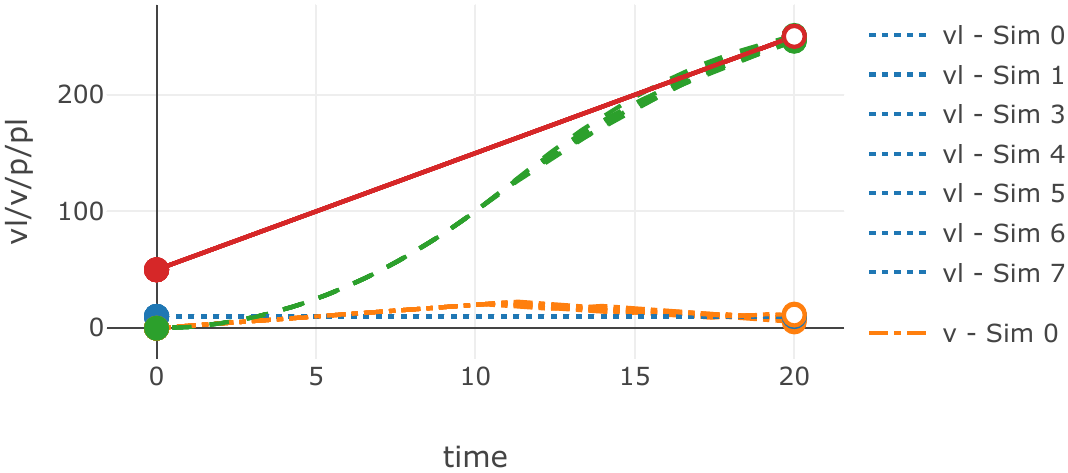}\\
  \wrap{\includegraphics[height=31mm]{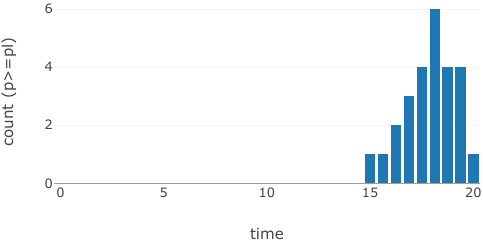}}\hfill%
  \wrap{\includegraphics[height=19mm]{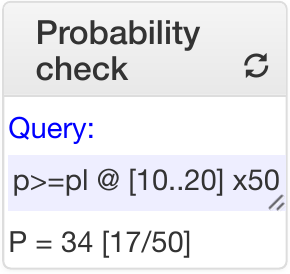}\\[3mm]}
  \caption{Multiple execution samples of different variants of 
          an adaptive cruise controller. Labels \textsf{p} and \textsf{v}
  denote the follower's position and velocity while \textsf{pl} and \textsf{vl}
indicate the leader's position and velocity. An histogram and a
probability checker that count how many times \textsf{p$\geq$pl} in 50 runs
with probabilistic waiting times.}
  \label{fig:ex-2.6}
\end{figure}

\begin{example}[Adaptive cruise controller]
        \label{acc}
        Lastly we consider a scenario in which a particle tries to follow
        another one as closely as possible and without crashing into it. For
        illustration purposes we consider that the following particle
        (henceforth, the follower) starts 50 meters behind the other one
        (henceforth, the leader) and that it is stationary.  It will be able to
        either accelerate or brake with forces \eg\ $2m/s^2$ and $-2m/s^2$,
        respectively, during 1 time unit each time.  The leader, on the other
        hand, starts with its velocity at $10m/s$ and cannot accelerate or
        brake.

        The follower's choice of whether to accelerate or brake each time is
        determined by checking whether, in case of choosing to accelerate
        during one time unit, a `safe braking distance' from the leader is
        maintained -- a braking distance is determined `safe' if the follower's braking
        trajectory does not intersect that of the leader.
        Technically this amounts to finding the roots of the pointwise
        difference of both trajectories (\ie\ finding the roots of a quadratic
        equation): the absence of roots amounts to the absence of intersections
        (see further details in~\cite{mendes24}). The overall idea of the
        scenario just described is encoded by the following program, where the
        operation
\lstinline|safe(p,v,pl,vl)|
informs whether roots were found or not.
\begin{lstlisting}[]
p := 0 ; v := 0 ; pl := 50 ; vl := 10 ;
while tt {
  if safe(p,v,pl,vl)
  then p' = v, v' =  2, pl' = vl, vl' = 0 for 1 
  else p' = v, v' = -2, pl' = vl, vl' = 0 for 1
}
\end{lstlisting}
        Let us now add some uncertainty to the leader: it will be able to
        uniformly take any acceleration in the range $[-1,1]$. This means that
        for \emph{complete safety}, the function
\lstinline|safe(p,v,pl,vl)|
        needs to be tweaked to assume the \emph{worst possible scenario}: \ie\
        while the follower's braking trajectory will be the same, the leader's
        trajectory is now assumed to be the one that results from choosing
        acceleration $-1 m/s^2$ (and not $0m/s^2$, as before). The resulting
        program is then as follows.
\begin{lstlisting}
p := 0 ; v := 0 ; pl := 50 ; vl := 10 ; a := 0;
while tt {
  a := unif(-1,1) ; 
  if safe(p,v,pl,vl)
  then p' = v, v' =  2, pl' = vl, vl' = a for 1 
  else p' = v, v' = -2, pl' = vl, vl' = a for 1
}
\end{lstlisting}
        Yet another option for introducing uncertainty is to consider the fact
        that the waiting times will be given by an exponential
        distribution, as follows.
\begin{lstlisting}
p := 0 ; v := 0 ; pl := 50 ; vl := 10 ; 
while tt {
  x := exp(lambda) ; x++ ; 
  if safe(p,v,pl,vl)
  then p' = v, v' =  2, pl' = vl, vl' = 0 for x
  else p' = v, v' = -2, pl' = vl, vl' = 0 for x
}
\end{lstlisting}
The function
\lstinline|safe(p,v,pl,vl)|
would then need to be tweaked again -- but remarkably now with no hope for
complete safety, as in theory \lstinline|x| can take any value from
$[1,\infty)$ and thus no worst-case scenario exists.

Figure~\ref{fig:ex-2.6} (top) presents an execution sample in which the system
is completely deterministic (\ie\ the leader's velocity is constant with
$100$\% certainty) and thus the follower gets as close as possible to the
leader. On the other hand, Figure~\ref{fig:ex-2.6} (middle) presents an
execution sample in which the follower assumes the worst-case scenario just
described and thus cannot get as close to the leader. Finally
Figure~\ref{fig:ex-2.6} (bottom plot) presents several execution samples
overlayed in which the original \lstinline|safe| function is used and the
respective durations are given by the exponential distribution 
\lstinline|1 + exp(8)|.  
It shows that, while collisions are improbable they do occur.  This
low probability of collision is depicted in the histogram at the bottom of 
Figure~\ref{fig:ex-2.6}, which counts how many times \lstinline|p >= pl| holds
in 50 runs over time (bottom-left) and anywhere in the interval [10,20]
(bottom-right).
\end{example}

\noindent
\textbf{Operational semantics.} 
The section's remainder is devoted to introducing an operational semantics for
the language -- not only such is a basis for formal reasoning about stochastic
hybrid programs it is also the engine of the interpreter that we have been
showcasing thus far.  In a nutshell, the semantics marries Kozen's operational
semantics for a probabilistic language~\cite{kozen79} with the semantics of
hybrid programs that was presented in~\cite{goncharov20,mendes24}.  We will
need some preliminaries.

We take the Hilbert cube $[0,1]^\omega$ as the source of randomness.
Operationally speaking this means that sampling will amount to drawing values
from an element of $[0,1]^\omega$ (a stream) that is fixed \emph{a priori}. For
example, sampling once will amount to taking the head of this element and
sampling $n$ times will amount to taking the respective prefix of size $n$.
Next we assume that the semantics of expressions \lstinline|e| and Boolean
conditions \lstinline|b| are given by partial maps $\sem{\text{\lstinline|e|}}:
\Reals^n \xrightharpoonup{\hspace{0.1cm}} \Reals$ and
$\sem{\text{\lstinline|b|}}: \Reals^n \xrightharpoonup{\hspace{0.1cm}}
\{\mathtt{tt},\mathtt{ff}\}$. These can be defined in the usual way. Now, since
we are in the context of imperative programming we will recur to the notion of
a store
$\sigma : \{$\lstinline[mathescape]|x$_1$, ... ,x$_n$|$\} \to \Reals$
(also known as memory or environment)~\cite{winskel93,reynolds98}.
It assigns a real number to any given variable in the language. For a store
$\sigma$, we will use the notation
$\sigma[$\lstinline|x |$ \mapsto v]$
to denote the store that is exactly like $\sigma$ except for the fact that
\lstinline|x| is now assigned value $v$. Finally we assume that any system of
differential equations in our language induces a continuous map $\phi : \Reals^n
\times \Rz \xrightharpoonup{\hspace{0.1cm}} \Reals^n$ -- which in our context
will be regarded as the respective solution. For simplicity  we 
denote an operation
\lstinline[mathescape]|x$_1$' $$= e$_1$,...,x$_n$' $$ = e$_n$ for e|
by the simpler expression
\lstinline[mathescape]|diff(e$_1$,...,e$_n$,e)|.

The rules of our semantics are then presented in Figure~\ref{fig:small-step}.
They dictate what the next computational step will be when evaluating a program
\lstinline|p| with initial store $\sigma$ w.r.t. time instant $t$. As alluded
before each evaluation is also associated with a source of randomness $s \in
[0,1]^\omega$ from which \lstinline|p| draws sampling results. Note from the
rules that such computational steps (coming out of $\text{\lstinline|p|} \sep
\sigma \sep t \sep s$) lead to one of three possible outcomes:
\emph{viz.} an error flag $\err$, an output store, or a resumption (\ie\ an
updated evaluation stack of programs, store, time instant, and source of
randomness) which can then be evaluated in the next step (the empty stack is
denoted by $\skp$).

A subtle feature of our semantics is that the time instant $t$ at the beginning
of the evaluation will tend to decrease along the computational steps
performed.  Intuitively this means that the evaluation is `moving forward in
time' until reaching the target time instant $t$ that we wish to evaluate. Most
notably when it detects that such time instant was surpassed it forces the
termination of the evaluation, even if the evaluation stack of programs is
currently non-empty.  Such is expressed by the rule \textbf{(seq-stop$^\to$)},
and is crucial for evaluating non-terminating programs, like the ones described
in Example~\ref{ctrw}, Example~\ref{ball}, Example~\ref{Einstein}, and
Example~\ref{acc}. We illustrate this feature next with a simple example, but
more details can also be found in~\cite{goncharov20}.
\begin{example}
\label{ex:stop}
Consider the following non-terminating program,
\begin{lstlisting}[mathescape]
x := 0 ; while tt $\{$ x++ ; wait 1 $\}$
\end{lstlisting}
        Although the loop involved does not terminate, one can always evaluate
        the program in a finite amount of steps for any given time instant.
        Let us see what happens, for example, at time instant $1 +
        \sfrac{1}{2}$.  First, for simplicity we denote the loop simply by
        \lstinline|p| and the store 
$\sigma : \{$ \lstinline|x| $\} \to \Reals$
that is defined as $\sigma($\lstinline|x|$) = v$ by \lstinline|x| $\mapsto v$.
We then deduce the following sequence of small-step transitions which arise
from the rules in Figure~\ref{fig:small-step}.
\begin{align*}
        \text{\lstinline|x:= 0 ; p|} \sep\ \sigma \sep\ 1 + \sfrac{1}{2}
        \sep\ s \to \> \>
        &
        \text{\lstinline|p|} \sep\ (\text{\lstinline|x|} \mapsto 0) \sep\
        1 + \sfrac{1}{2} \sep\ s
        \\
        \to \> \>
        &
        \text{\lstinline|x++ ; wait 1 ; p|} \sep\ (\text{\lstinline|x|} \mapsto 0)
        \sep\
        1 + \sfrac{1}{2} \sep\ s
        \\
        \to \> \>
        &
        \text{\lstinline|wait 1 ; p|} \sep\ (\text{\lstinline|x|} \mapsto 1)
        \sep\
        1 + \sfrac{1}{2} \sep\ s
        \\
        \to \> \>
        &
        \text{\lstinline|p|} \sep\ (\text{\lstinline|x|} \mapsto 1)
        \sep\
        \sfrac{1}{2} \sep\ s
        \\
        \to \> \>
        &
        \text{\lstinline|x++ ; wait 1 ; p|} \sep\ (\text{\lstinline|x|} \mapsto 1)
        \sep\
        \sfrac{1}{2} \sep\ s
        \\
        \to \> \>
        &
        \text{\lstinline|wait 1 ; p|} \sep\ (\text{\lstinline|x|} \mapsto 2)
        \sep\
        \sfrac{1}{2} \sep\ s
        \\
        \to \> \>
        &
        \text{\lstinline|x|} \mapsto 2
\end{align*}
The last transition arises precisely due to rules $\textbf{(diff-stop$^\to$)}$
and $\textbf{(seq-stop$^\to$)}$, which are applicable due to $\sfrac{1}{2}$
being \emph{strictly} smaller than the wait call's duration (\emph{viz.} 1).
Henceforth we will call events such as the one just described \emph{time-based
terminations}, in order to distinguish from those ordinary terminations that
originate from emptying the program evaluation stack. We will see later on that
this subtle aspect can be neatly handled in the denotational context via an
exception monad.

\end{example}

Let us briefly mention how our interpreter uses this semantics to provide
(overlayed) execution samples of a given stochastic hybrid program
\lstinline|p|. The basic idea is simple: we first generate an entropy source \ie\
a sample $s
\in [0,1]^\omega$ and then use it to compute the execution chain
$\text{\lstinline|p|} \sep\ \sigma \sep\ t \sep\ s \ssto \dots$ for multiple time
instants $t$, corresponding to different snapshots of the program's 
behavioural trajectory. In order to obtain overlayed execution samples one just repeats
this process multiple times, \ie\ with different samples.

We now introduce a big-step operational semantics in
Figure~\ref{fig:big-step2}, which abstracts from intermediate computational
steps in the context of the small-step variant. Although in programming theory
big-step semantics have multiple applications~\cite{winskel93,reynolds98}, here
we use it just to connect the small-step variant to the denotational
counterpart (in Section~\ref{sec:denot}). In other words, we use the big-step
variant as a midpoint between small-step and denotational semantics. Since this
big-step semantics is based on the same ideas as the small-step one, we skip
its explanation.

We conclude the section by showing that the small-step and big-step semantics
agree, in the sense that they give rise to the same input-output relation.
Technically we will need to factor in a form of transitive closure from the
small-step relation, as follows. We call `terminal' those tuples arising from
steps (in the small-step semantics) that are of the form $\skp \sep\ \sigma
\sep\ t \sep\ s$, or $\sigma$, or $\err$.  Then we build an input-output
relation ($\Rightarrow$) from the small-step semantics in the way detailed in
Figure~\ref{fig:big-step}. Finally,

\begin{theorem}
        For every program \lstinline|p|, store $\sigma$, time instant $t$, and
        source of randomness $s$, we have the following following equivalence: 
        \[
                \text{\lstinline|p|} \sep\ \sigma \sep\ t \sep\ s
                \bsto v 
                \quad
                \text{ iff }
                \quad
                \text{\lstinline|p|} \sep\ \sigma \sep\ t \sep\ s
                ~\Rightarrow~ v 
        \]
\end{theorem} 
\begin{proof}
The right-to-left direction follows by induction on the length of
 small-step reduction sequences and Lemma~\ref{lem:progress}.
The left-to-right direction follows by induction over big-step derivations.
\end{proof}

\begin{lemma}\label{lem:progress}
For a program \lstinline|p|, a store $\sigma$, time instant $t$, and a source
of randomness $s$, the following holds:
if\/ $\text{\lstinline|p|} \sep\ \sigma \sep\ t \sep\ s 
                \ssto \text{\lstinline|p|}' \sep\ \sigma'\sep\
                t' \sep\ s'$
                and\/ $\text{ \lstinline|p|}'\sep\ \sigma' \sep\ t' \sep\ s' 
                \bsto v$ then
                $\text{\lstinline|p|} \sep\ \sigma \sep\ t \sep\ s
                \bsto v$.
\end{lemma}

\begin{proof} 
        Follows straightforwardly by induction over small-step
        derivations. 
\end{proof}
\begin{figure*}[h!]
\begin{minipage}{1\textwidth}
\begin{flalign*}
\textbf{(asg-rnd$^\to$)}
&&
\hspace{-3cm}
\text{\lstinline|x := unif(0,1)|}
\sep\ \sigma \sep\ t  \sep\ (h:s)
\ssto
\skp \sep\ \sigma[\text{\lstinline|x|} \mapsto h] \sep\ t \sep\ s
&&
\end{flalign*} 

\nline

\begin{flalign*}
\textbf{(asg$^\to$)}
&&
\text{\lstinline|x := e|} \sep\ \sigma \sep\ t \sep\ s
\ssto
\skp \sep\ \sigma[\text{\lstinline|x|} \mapsto 
\sem{\text{\lstinline|e|}}(\sigma) ] \sep\ t  \sep\ s
&&
\prem{\sem{\text{\lstinline|e|}}(\sigma) \text{ defined}}
\end{flalign*} 

\nline

\begin{flalign*}
\textbf{(asg-err$^\to$)}
&&
\text{\lstinline|x := e|} \sep\ \sigma \sep\ t \sep\ s
\ssto \err 
&&
\prem{\sem{\text{\lstinline|e|}}(\sigma) \text{ undefined}}
\end{flalign*} 

\nline

\begin{flalign*}
\textbf{(diff-stop$^\to$)}
  &&
  \text{\lstinline[mathescape]|diff(e$_1$,...,e$_n$,e)|}
  \sep\ \sigma \sep\ t \sep\ s
  \ssto
  \phi(\sigma,t) 
  &&
  \prem{\sem{\text{\lstinline|e|}}(\sigma) > t}
\end{flalign*} \nline
\begin{flalign*}
\textbf{(diff-skip$^\to$)}
&& 
    \text{\lstinline[mathescape]|diff(e$_1$,...,e$_n$,e)|}
    \sep\ \sigma \sep\ t \sep\ s
    \ssto
    \skp \sep\ 
    \phi(\sigma, \sem{\text{\lstinline|e|}}(\sigma)) 
    \sep\ t - \sem{\text{\lstinline|e|}}(\sigma)
    \sep\ s
&&
\prem{0 \leq \sem{\text{\lstinline|e|}}(\sigma) \leq t}
\end{flalign*} \nline
\begin{flalign*}
\textbf{(diff-err$^\to$)}
&& 
  \text{\lstinline[mathescape]|diff(e$_1$,...,e$_n$,e)|}
  \sep\ \sigma \sep\ t \sep\ s
  \ssto \err 
&&
\prem{\sem{\text{\lstinline|e|}}(\sigma) < 0 \text{ or }
\sem{\text{\lstinline|e|}}(\sigma) \text{ undefined}}
\end{flalign*} \nline
\begin{flalign*}
\textbf{(if-true$^\to$)}
&&
\text{\lstinline|if b then p else q|}
\sep\ \sigma \sep\ t \sep\ s \ssto 
\text{\lstinline|p|} \sep\ \sigma \sep\ t \sep\ s 
&&
\prem{\sem{\text{\lstinline|b|}}(\sigma) =\mathtt{tt}}
\end{flalign*} \nline
\begin{flalign*}
\textbf{(if-false$^\to$)}
&&
\text{\lstinline|if b then p else q|}
\sep\ \sigma \sep\ t \sep\ s \ssto 
\text{\lstinline|q|} \sep\ \sigma \sep\ t \sep\ s 
&&
\prem{\sem{\text{\lstinline|q|}}(\sigma)=\mathtt{ff}}
\end{flalign*} \nline
\begin{flalign*}
\textbf{(if-err$^\to$)}
&&
\text{\lstinline|if b then p else q|}
\sep\ \sigma \sep\ t \sep\ s \ssto \err
&&
\prem{\sem{\text{\lstinline|b|}}(\sigma) \text{ undefined}}
\end{flalign*} \nline
\begin{flalign*}
\textbf{(wh-true$^\to$)}
&&
\text{\lstinline[mathescape]|while b do $\> \> \{ \>$ p $ \> \}$|}
\sep\ \sigma \sep\ t \sep\ s \ssto 
\text{\lstinline[mathescape]|p ; while b do  $\> \> \{ \>$ p  $ \> \}$|}
\sep\ \sigma \sep\ t \sep\ s 
&&
\prem{\sem{\text{\lstinline|b|}}(\sigma)=\mathtt{tt}}
\end{flalign*} \nline
\begin{flalign*}
\textbf{(wh-false$^\to$)}
&&
\text{\lstinline[mathescape]|while b do $\> \> \{ \>$ p $ \> \}$|}
\sep\ \sigma \sep\ t \sep\ s \ssto \skp \sep\ \sigma \sep\ t
\sep\ s 
&&
\prem{\sem{\text{\lstinline|b|}}(\sigma)=\mathtt{ff}}
\end{flalign*}\nline 
\begin{flalign*}
\textbf{(wh-err$^\to$)}
&&
\text{\lstinline[mathescape]|while b do $\> \> \{ \>$ p $ \> \}$|}
\sep\ \sigma \sep\ t \sep\ s \ssto \err
&&
\prem{\sem{\text{\lstinline|b|}}(\sigma) \text{ undefined} }
\end{flalign*} 
\begin{flalign*}
\lrule{(seq-stop$^\to$)}{\text{\lstinline|p|} \sep\ \sigma \sep\ t 
\sep\ s \ssto \sigma'}{
\text{\lstinline|p ; q|} \sep\ \sigma \sep\ t \sep\ s  \ssto \sigma'}
&&
\lrule{(seq-skip$^\to$)}{\text{\lstinline|p|} \sep\ \sigma \sep\ t 
        \sep\ s \ssto \skp \sep\ \sigma' \sep\ t' \sep\ s'}{
        \text{\lstinline|p ; q|} \sep\ \sigma \sep\ t \sep\ s 
        \ssto \text{\lstinline|q|} \sep\ \sigma' \sep\ t'
        \sep\ s'}
\end{flalign*}
\begin{flalign*}
\lrule{(seq-err$^\to$)}{\text{\lstinline|p|} \sep\ \sigma \sep\ t \sep\ s 
\ssto \err}{
\text{\lstinline|p ; q|} \sep\ \sigma \sep\ t \sep\ s  \ssto \err }
\hspace{0.9cm}
&&
\lrule{(seq$^\to$)}{
\text{\lstinline|p|} \sep\ \sigma \sep\ t \sep\ s 
\ssto \text{\lstinline|p|}' \sep\ \sigma' \sep\ t' \sep\ s'}{
\text{\lstinline|p ; q|} \sep\ \sigma \sep\ t \sep\ s 
\ssto \text{\lstinline|p|}' \text{ \lstinline| ; q|} \sep\ \sigma' \sep\ t' \sep\ s'} 
\end{flalign*} 
  \end{minipage}
  \caption{Small-step operational semantics.}
  \label{fig:small-step}
\end{figure*}

\begin{figure*}
\begin{minipage}{1\textwidth}
\begin{flalign*} 
        \lrule{(asg-rnd)}{
        }{
                \text{\lstinline{x := unif(0,1)}} \sep\
                \sigma \sep\ t \sep\ (h:s)
                \bsto \skp \sep\ \sigma[\text{\lstinline{x}} \mapsto h] \sep\ t \sep\ s
        }
\end{flalign*}
\\[-12pt]
\begin{flalign*}
        \lrule{(asg-skip)}{
                \sem{\text{\lstinline{e}}}(\sigma) \text{ defined }
        }{
                \text{\lstinline{x := e}} \sep\ \sigma 
                \sep\ t \sep\ s
                \bsto
                \skp \sep\ 
                \sigma[\text{\lstinline{x}} \mapsto 
                \sem{\text{\lstinline{e}}}(\sigma)] \sep\ t \sep\ s
        }
&&
        \lrule{(asg-err)}{ 
                \sem{\text{\lstinline{e}}}(\sigma) \text{ undefined }
        }{
                \text{\lstinline{x := e}}
                \sep\ \sigma \sep\ t \sep\ s
                \bsto
                \err
        }
\end{flalign*} 
\\[-12pt]
\begin{flalign*}
        \lrule{(diff-skip)}{
                0 \leq \sem{\text{\lstinline{e}}}(\sigma) \leq t
        }{
\text{\lstinline[mathescape]|diff(e$_1$,...,e$_n$,,e)|} 
                \sep\ \sigma \sep\ t \sep\ s
                \bsto
                \skp \sep\ \phi(\sigma, \sem{\text{\lstinline{e}}}(\sigma))
                \sep\ t - \sem{\text{\lstinline{e}}}(\sigma) \sep\ s
        }
\end{flalign*} 
\\[-12pt]
\begin{flalign*}
        \lrule{(diff-stop)}{
                \sem{\text{\lstinline|e|}}(\sigma) > t
        }{
\text{\lstinline[mathescape]|diff(e$_1$,...,e$_n$,,e)|} 
                \sep\ \sigma \sep\ t \sep\ s
                \bsto
                \phi(\sigma, t)
        }
&&
\lrule{(diff-err)}{\sem{\text{\lstinline|e|}}(\sigma) < 0 \text{ or } 
\sem{\text{\lstinline|e|}}(\sigma) \text{ undefined }}{
\text{\lstinline[mathescape]|diff(e$_1$,...,e$_n$,,e)|} 
    \bsto
    \err
}
\end{flalign*} 
\\[-12pt]
\begin{flalign*}
        \lrule{(seq-skip)}{
                \text{\lstinline|p|} \sep\
                \sigma \sep\ t \sep\ s 
                \bsto \skp \sep\ \sigma' \sep\ t' \sep\ s' \qquad 
                \text{\lstinline|q|}
                \sep\ \sigma' \sep\ t' \sep\ s' \bsto v
        }{
                \text{\lstinline|p ; q|}
                \sep\ \sigma \sep\ t \sep\ s \bsto v
        }
\end{flalign*} 
\\[-12pt]
\begin{flalign*}
        \lrule{(seq-stop)}{
                \text{\lstinline|p|}
                \sep\
                \sigma \sep\ t \sep\ s \bsto \sigma'
        }{
                \text{\lstinline|p ; q|}
                \sep\ \sigma \sep\ 
                t \sep\ s \bsto \sigma'
        }
&&
        \lrule{(seq-err)}{
                \text{\lstinline|p|}
                \sep\ \sigma \sep\ t \sep\ s \bsto \err
        }{
                \text{\lstinline|p ; q|}
                \sep\ \sigma
                \sep\ t \sep\ s \bsto \err
        }
\end{flalign*} 
\\[-12pt]
\begin{flalign*}
        \lrule{(if-true)}{
                \sem{\text{\lstinline|b|}}(\sigma)=\mathtt{tt}
                \qquad 
                \text{\lstinline|p|}
                \sep\ \sigma \sep\ t \sep\ s \bsto v
        }{
                { \text{\lstinline|if b then p else q|} } 
                \sep\ \sigma \sep\ t \sep\ s \bsto v
        }
\end{flalign*} 
\\[-12pt]
\begin{flalign*}
        \lrule{(if-false)}{
                \sem{\text{\lstinline|b|}}(\sigma)=\mathtt{ff}
                \qquad 
                \text{\lstinline|q|}
                \sep\ \sigma \sep\ t \sep\ s \bsto v
        }{
                \text{\lstinline|if b then p else q|}
                \sep\ \sigma \sep\ t \sep\ s \bsto v
        }
&&
        \lrule{(if-err)}{
                \sem{\text{\lstinline|b|}}(\sigma) \text{ undefined }
        }{
                \text{\lstinline|if b then p else q|}
                \sep\ \sigma \sep\ t \sep\ s \bsto \err
        }
\end{flalign*} 
\\[-12pt]
\begin{flalign*}
        \lrule{(wh-true)}{
                \sem{\text{\lstinline|b|}}(\sigma) = 
                \mathtt{tt}\qquad 
                \text{\lstinline[mathescape]|p ; while b do $\> \{\>$ p $\> \}$|}
                \sep\ \sigma \sep\ t \sep\ s \bsto v 
        }{
                \text{\lstinline[mathescape]|while b do $\> \{\>$ p $\> \}$|}
                \sep\ \sigma \sep\ t \sep\ s
                \bsto 
                v
        }
\end{flalign*}
\\[-12pt]
\begin{flalign*}
        \lrule{(wh-false)}{
                \sem{\text{\lstinline|b|}}(\sigma) = \mathtt{ff}
        }{
                \text{\lstinline[mathescape]|while b do $\> \{\>$ p $\> \}$|}
                \sep\ \sigma \sep\ t \sep\ s
                \bsto 
                \skp \sep\ \sigma \sep\ t \sep\ s 
        }
&&
        \lrule{(wh-err)}{
                \sem{\text{\lstinline|b|}}(\sigma) \text{ undefined }
        }{
                \text{\lstinline[mathescape]|while b do $\> \{\>$ p $\> \}$|}
                \sep\ \sigma \sep\ t \sep\ s
                \bsto 
                \err
        }
\end{flalign*}
\end{minipage}
\caption{Big-step operational semantics.}
\label{fig:big-step2}
\end{figure*}

\begin{figure*}
\begin{minipage}{1\textwidth}
\begin{flalign*}
        \infer[(v \text{ terminal})]{
        \text{\lstinline|p|} \sep\ \sigma \sep\ t \sep\ s ~\Rightarrow~ v}
        {\text{\lstinline|p|} \sep\ \sigma \sep\ t \sep\ s \ssto v}
        \hspace{2.5cm}
        \infer{
        \text{\lstinline|p|} \sep\ \sigma \sep\ t \sep\ s \Rightarrow v 
        }{ \text{\lstinline|p|} \sep\ \sigma \sep\ t \sep\ s\ssto
        \text{\lstinline|p|}' \sep\ \sigma' \sep\ t' \sep\ s'
        \qquad 
        \text{\lstinline|p|}' \sep\ \sigma' \sep\ t' \sep\ s' \Rightarrow v 
}
\end{flalign*}
\end{minipage}
\caption{Big-step semantics via the small-step relation.}
\label{fig:big-step}
\end{figure*}

\section{Measure theory}
\label{sec:bck}
This section briefly recalls a series of results about measure
theory~\cite{aliprantis06,panangaden09,dudley02}, focus being on those that
form the backbone of the semantics described in Section~\ref{sec:denot}.

Our main working category will be $\Meas$, \ie\ that of measurable spaces and
measurable functions. Recall that it has both (infinite) products and
coproducts~\cite[Section 21]{cats}.  Recall as well that it is distributive,
\ie\ for all measurable spaces $X,Y,Z$ there exists a  certain isomorphism,
\[
        \dist : X \times (Y +Z) \to X \times Y + X \times Z
\]
Now, let $\CTop$ be the category of topological spaces and continuous maps, and
recall that it has (infinite) products and coproducts as well~\cite[Section
21]{cats}.  There exists a functor $B : \CTop \to \Meas$ that sends any given
topological space to the measurable space with the same carrier and equipped
with the respective Borel $\sigma$-algebra~\cite[Section 4.4]{aliprantis06}. In
particular when treating a subset of real numbers as a measurable space we will
be tacitly referring to the respective Borel $\sigma$-algebra. It is well-known
that $B$ preserves finite products of second-countable topological
spaces~\cite[Definition 6.3.7]{larrecq13}, which is the case for example of
Polish spaces (see \cite[Chapter 3]{aliprantis06} and \cite[Theorem
6.3.44]{larrecq13}). This property is key for our semantics: it will allow us
to treat solutions of systems of differential equations -- which are continuous
functions and thus live in $\CTop$ -- as measurable functions. Further details
about this crucial aspect are available in the following section.

We proceed by briefly recalling basic results about measures -- a more detailed
description is available for example in \cite[Chapter 10]{aliprantis06},
\cite[Chapter 1]{barthe20}, and~\cite[Chapter 2]{panangaden09}.

\begin{definition} For a measurable space $(X,\Sigma_X)$ a signed measure is a
        function $\mu : \Sigma_X \to \Reals$ such that 
        $\mu(\emptyset) = 0$ and
        moreover it is $\sigma$-additive, \ie\ 
        \[
        \mu \left (\bigcup_{i =1}^{\infty} U_i \right ) = \sum_{i = 1}^{\infty}
        \mu(U_i) 
        \] 
where $(U_i)_{i \in \omega}$ is any family of pairwise disjoint measurable
sets. A signed measure will be simply called a measure if $\mu(U) \geq 0$ for
all measurable subsets $U\subseteq X$. As usual a measure is called a
subdistribution if $\mu(X) \leq 1$ and a distribution if $\mu(X) = 1$.
\end{definition}

For a measurable space $X$ the set of signed measures $\Measu(X)$ forms a
vector space via pointwise extension. It also forms a normed space when
equipped with the total variation norm,
\[
        \lVert \mu \rVert = 
        \sup \left \{ \sum_{i = 1}^n \, | \mu(U_i) | \mid
             \{ U_1, \dots, U_n \} \text{ 
             measurable partition of $X$}
        \right \}
\]
In particular for a measure $\mu$ we have $\lVert \mu \rVert = \mu(X)$. Note
that $\Measu(X)$ is also a Banach space by virtue of the reals numbers forming
a Banach space, specifically the limit of a Cauchy sequence is built via
pointwise extension. 

Take subdistributions $\mu \in \Measu(X)$ and $\nu \in \Measu(Y)$.  There
exists the so-called \emph{tensor} or \emph{product measure} $\mu \otimes \nu
\in \Measu(X \times Y)$, which is defined by the equation $\mu \otimes \nu (U
\times V) = \mu(U) \nu(V)$ on all measurable rectangles $U \times V \in
\Sigma_{X \times Y}$. Specifically the latter extends standardly to all
measurable sets by an appeal to Carathéodory's extension and moreover the
extension is unique (see \eg\ \cite[Lemma 10.33]{aliprantis06}).  Another
useful fact is that for any subdistributions $\mu,\nu \in \Measu(X)$ and $\rho
\in \Measu(Y)$ the equation below holds. 
\[
        (\mu + \nu) \otimes \rho = \mu \otimes \rho + \nu \otimes \rho
\]
The product measure construction just described also applies to countable
families of distributions $(\mu_i)_{i \in \omega}$ in $\Measu(X)$.

Next, for a measurable space $X$ we will denote by $\Giry(X)$ the set of
subdistributions.  The construct $\Giry(-)$ thus defined forms the Giry monad
in $\Meas$ when every $\Giry(X)$ is equipped with the $\sigma$-algebra
generated by the evaluation maps,
\[
        \eval_U : \Giry(X) \to \Reals
        \hspace{1cm} 
        \mu \mapsto \mu(U)
        \hspace{1cm}
        (U \subseteq X \text{ measurable})
\]
\cite{panangaden99}. 
The respective Kleisli morphisms $f : X \to \Giry(Y)$ are typically called
Markov kernels and their Kleisli extension $f^\star : \Giry(X) \to \Giry(Y)$ is
given by Lebesgue integration~\cite[Chapter 11]{aliprantis06},
\[
        f^\star (\mu) = U \mapsto \int_{x \in X} f(x)(U) \, d \mu(x)
\]
For every measurable space $X$ the unit $\delta : X \to \Giry(X)$ of this monad
is given by the Dirac delta $\delta_x \in \Giry(X)$ with $x \in X$, \ie\
\[
        \delta_x(U) = 
        \begin{cases}
                1 & \text{ if } x \in U \\
                0 & \text{ otherwise }
        \end{cases}
\]
In particular the functorial action of $\Giry: \Meas \to \Meas$ is the
pushforward measure operation. We will often abuse notation by denoting a
linear combination $\sum_i p_i \cdot \delta_{x_i}$ simply by $\sum_i p_i \cdot
x_i$ with $x_i \in X$. 

Let us  recall useful properties about the Giry monad. First it is commutative
when equipped with the double-strength operation $\Giry(X) \times \Giry(Y) \to
\Giry(X \times Y)$ defined via the product measure~\cite{sato18}.  The fact
that such operation is measurable follows from~\cite[Lemma 4.11]{aliprantis06}.
Second for any bounded measurable map $f : X \to \Reals$ and measures $\mu,\nu
\in \Giry(X)$ Lebesgue integration satisfies the conditions, 
\[ 
        \int f d (\mu +
        \nu) = \int f d \mu\  + \int g d \nu \hspace{0.5cm} 
        \int f d (s \cdot \mu)
        = s \cdot \left (\int f d \mu \right ) 
\] 
for any scalar $s \in \Reals$.  Thus we immediately conclude that the Kleisli
extension of a Markov kernel will always be linear. Third if the codomain of
$f$ restricts to $[0,1]$ we obtain,
\[ 
        \int f d \mu \leq \mu (X)
        = \lVert \mu \rVert
\] 
This entails that the Kleisli extension of a Markov kernel will be bounded and
thus continuous (even contractive) w.r.t. the metric induced by the total
variation norm. This provides a number of tools from functional analysis.  For
example one can analyse how $f^\star$ acts on a measure $\mu$ by a series of
approximations $\mu_n$ to $\mu$.  Not only this, the set of maps
$\Meas(\Giry(X), \Giry(Y))$ can be equipped with the metric induced by the
operator norm,
\[
        \lVert T \rVert = \sup
        \left \{
                \lVert T (\mu) \rVert \mid \mu \in \Giry(X),  \lVert \mu \rVert \leq 1
        \right \}
\]
which moves us beyond classical program equivalence by allowing to compare
programs in terms of distances and not just equality (see for example
\cite{dahlqvist23}).  More details about this last aspect will be given later
on.

Another useful fact is that for a given measure $\mu \in \Giry(X)$ any
measurable subset $U \subseteq X$ gives rise to a new measure $\mu(U \cap - )$.
Moreover for any bounded measurable map $f : X \to \Reals$ and measure $\mu\in
\Giry(X)$ we obtain,
\[
        \int_X f  d \mu 
         = \int_U  f  d \mu(U \cap -) + 
        \int_{\overline{U}} f d \mu(\overline{U} \cap -)
\]
where $\overline{U}$ represents the complement of $U$. Yet  another useful
property of the Giry monad is that for every measurable space $X$ the space
$\Giry(X)$ inherits the usual order on the real numbers, via pointwise
extension. What is more, the induced order has a bottom element (the zero-mass
measure) and it is $\omega$-complete, by virtue of the completeness property of
the real numbers. Remarkably, an $\omega$-increasing sequence of measures
$(\mu_n)_{n \in \omega}$ in $\Giry(X)$ is Cauchy and $\sup_{n \in \omega} \mu_n
= \lim_{n \to \infty} \mu_n$, thanks to the monotone convergence theorem (see
\eg\ \cite[Theorem 11.18]{aliprantis06} or \cite[Theorem 3.6]{panangaden09}).
This is helpful to jump between domain theory and functional analysis whenever
necessary. 

Next, observe that the aforementioned order extends to Markov kernels via
pointwise extension. It has a bottom element (the map constant on the zero-mass
measure) and it is $\omega$-complete. The last property follows directly from
the definition of the $\sigma$-algebra of $\Giry(X)$ for every $X$ and
from the fact that the pointwise supremum of real-valued measurable functions
is measurable. Also, it follows from the monotone convergence theorem that the
equation,
\[
        (\sup_{i \in \omega} f_i)^\star = \sup_{i \in \omega} f_i^\star
\]
holds for any increasing sequence $(f_i)_{i \in \omega}$ of Markov kernels.
This will be crucial for the interpretation of while-loops in the following
section.
Finally it follows from the fact that multiplication preserves suprema that,
\[
        (\sup_{i \in \omega} \mu_i) \otimes \nu = \sup_{i \in \omega} \mu_i \otimes \nu
\]
for any increasing sequence of subdistributions $(\mu_i)_{i \in \omega}$ in
$\Giry(X)$ and $\nu \in \Giry(Y)$.

\section{Denotational semantics}
\label{sec:denot}

We now introduce a denotational, measure-theoretic semantics for our stochastic
language. In a nutshell, it extends Kozen's well-known probabilistic
semantics~\cite{kozen79} with a mechanism for handling the time-based
terminations that were described in Section~\ref{sec:lng}. The extension boils
down to the following categorical construction.  

Any object $E$ in a category $\catC$ with binary coproducts induces a monad $E
+ (-)$ which intuitively gives semantics to exception
handling~\cite{moggi:1991}. It follows from the universal property of
coproducts that any monad $T$ in $\catC$ combines with $E + (-)$. In other
words we have a new monad $T \otimes E$ which handles at the same time effects
arising from $T$ and exceptions. Concretely the functorial action of this new
monad is given by $T(E + (-))$, the unit $\eta^{T \otimes E}$ by the
composition $\eta^T \comp \inr : X \to T(E + X)$, and the Kleisli lifting
$(-)^{\star^{T \otimes E}}$ by the equation,
\begin{align}
        \label{eq:combined}
        f^{\star^{T \otimes E}} = [\eta^T \comp \inl, f]^{\star^T}
\end{align}
Time-based terminations will be handled precisely via one such monad $\Giry
\otimes E$ in $\Meas$ -- in other words these terminations are technically seen
as exceptions, in the sense that they also inhibit the execution of subsequent
computations and instead are merely propagated forward along the evaluation.
Specifically the denotation $\sem{\text{\lstinline|p|}}$ of a program
\lstinline|p| will be a Markov kernel $X \to \Giry(E + X)$ in which elements of
$E$ denote time-based terminations. The space $X$ will be in particular the
product $\Reals^n \times \Rz$ whilst $E$ will be $\Reals^n$ (thus analogously
to Section~\ref{sec:lng}, $n$ is the cardinality of our stock of variables and
possible outputs are either elements of $\Reals^n \times \Rz$ or
$\Reals^n$).  Consequently, for every $(\sigma,t) \in \Reals^n \times \Rz$ we
have $\sem{\text{\lstinline|p|}}(\sigma,t)$ as a subdistribution which assigns
probabilities to the outputs of \lstinline|p| w.r.t. time instant $t$ and
initial state $\sigma$. For any given initial state $\sigma \in \Reals^n$, one
can also see a denotation $\sem{\text{\lstinline|p|}}$ as inducing a Markov
process $\sem{\text{\lstinline|p|}}(\sigma,-)$ which intuitively means that the
subdistribution of outputs evolves over time. Note as well that the possibility
of the total mass of $\sem{\text{\lstinline|p|}}(\sigma,t)$ being strictly
lower than $1$ for a given input $(\sigma,t) \in  \Reals^n \times \Rz$ reflects
the possibility of divergence and/or errors in the evaluation of expressions
and Boolean conditions.

Lastly in order to interpret while-loops, we observe that the combined monad
$\Giry \otimes E$ inherits the order of $\Giry$, and moreover,
\begin{align}
        \label{eq:scott}
        (\sup_{i \in \omega} f_i)^{\star^{\Giry \otimes E}} 
        = \sup_{i \in \omega} f_i^{\star^{\Giry \otimes E}}
\end{align}
for any increasing sequence $(f_i)_{i \in \omega}$ of Markov kernels. This last
equation follows from the Scott-continuity of co-pairing on its second
argument. 

We are finally ready to introduce our denotational semantics. It is defined in
Figure~\ref{denot_sem2}, via induction on the syntactic structure of programs.
\begin{figure*}
\begin{minipage}{1\textwidth}
\begin{align*}
        \sem{ \text{\lstinline[mathescape]|diff(e$_1$,...,e$_n$,e)|} } & =
        (\sigma,t) \mapsto
        \begin{cases}
                1 \cdot \phi(\sigma,t)
                & \text{ if } \sem{\text{\lstinline{e}}}(\sigma) > t
                \\
                1 \cdot (\phi(\sigma,\sem{\text{\lstinline{e}}}(\sigma)),  
                t - \sem{\text{\lstinline{e}}}(\sigma))
                & \text{ if } 0 \leq \sem{\text{\lstinline{e}}}(\sigma) \leq t
                \\
                0 
                & \text{ otherwise }
        \end{cases}
        \\
        \sem{ \text{\lstinline{x := e}} } & = (\sigma,t) \mapsto 
        \begin{cases}
                1 \cdot (\sigma[\text{\lstinline{x}} 
                \mapsto \sem{\text{\lstinline{e}}}(\sigma)],t) & \text{ if }
                \sem{\text{\lstinline{e}}}(\sigma) \text{ is well-defined }
                \\
                0 & \text{ otherwise }
        \end{cases}
        \\
        \sem{ \text{\lstinline[mathescape]|x$_i$ := unif(0,1)|} } & = (\sigma,t)
        \mapsto
        \sigma[\text{\lstinline[mathescape]|x$_i$|} 
        \mapsto \lambda] \otimes ( 1 \cdot t )
        \\
        \sem{ \text{\lstinline{if b then p else q}} }  & = (\sigma,t) \mapsto
        \begin{cases}
                \sem{\text{\lstinline|p|}}(\sigma,t)
                &
                \text{ if } \sem{\text{\lstinline|b|}}(\sigma) = \mathtt{tt}
                \\
                \sem{\text{\lstinline|q|}}(\sigma,t)
                &
                \text{ if } \sem{\text{\lstinline|b|}}(\sigma) = \mathtt{ff}
                \\
                0 & \text{ otherwise}
        \end{cases}
        \\
        \sem{\text{\lstinline{p ; q}}} & = 
        \sem{ \text{\lstinline{q}} }^{\star}
                \comp
                \sem{ \text{\lstinline{p}} }
        \\
        \sem{ \text{\lstinline{while b do p}} } & = 
        \mathrm{lfp} \left (
                k \mapsto 
                (\sigma,t) \mapsto
                \begin{cases}
                        k^\star \comp \sem{\text{\lstinline|p|}} (\sigma,t)
                        & \text{ if } \sem{\text{\lstinline|b|}} (\sigma) =
                        \text{\lstinline|tt|}
                        \\
                        1 \cdot (\sigma,t)
                        & \text{ if } \sem{\text{\lstinline|b|}}(\sigma) =
                        \text{\lstinline|ff|}
                        \\
                        0
                        & \text{ otherwise } 
                \end{cases}
               \right )
\end{align*}
\end{minipage}
\caption{Denotational semantics.}
\label{denot_sem2}
\end{figure*}
It assumes, as usual, that the semantics of expressions \lstinline|e| and
Boolean conditions \lstinline|b| are given by \emph{measurable} partial maps
$\sem{\text{\lstinline|e|}}: \Reals^n \xrightharpoonup{\hspace{0.1cm}} \Reals$
and $\sem{\text{\lstinline|b|}}: \Reals^n \xrightharpoonup{\hspace{0.1cm}}
\{\mathtt{tt},\mathtt{ff}\}$. The measurability of each interpretation clause
in Figure~\ref{denot_sem2} is then straightforward to verify. Indeed, the only
somewhat complicated case is the first clause, which crucially relies on two
related properties. First, the fact that every continuous map $\phi : \Reals^n
\times \Rz \to \Reals^n$ is
measurable as a map $B(\phi) : B(\Reals)^n \times B(\Rz) \to B(\Reals)^n$ (which we commented on in the last section).
Second, by an analogous reasoning, the fact that the subtraction $\Reals
\times_\Meas \Reals \to \Reals$ map is measurable. This entails in particular
that the strictly greater relation $(>)$ is measurable as a function $\Reals
\times_\Meas \Reals \to 1 + 1$, by virtue of $(-\infty,0)$ being a measurable
subset of $\Reals$. Next, observe our slight abuse of notation in
$\sigma[\text{\lstinline[mathescape]|x$_i$|} \mapsto \lambda]$ which
abbreviates the product measure,
\[
        1 \cdot \sigma(\text{\lstinline[mathescape]|x$_1$|})
        \, \otimes \,
        \dots
        \, \otimes \,
        1 \cdot \sigma(\text{\lstinline[mathescape]|x$_{i-1}$|})
        \, \otimes \,
        \lambda
        \, \otimes \, 
        1 \cdot \sigma(\text{\lstinline[mathescape]|x$_{i+1}$|})
        \, \otimes \,
        \dots
        \, \otimes \, 
        1 \cdot \sigma(\text{\lstinline[mathescape]|x$_{n}$|})
\]
where $\lambda$ is the uniform distribution on $[0,1]$.  Finally the last
clause interprets while-loops via Kleene's least fixpoint construction. The
fact that the map from which we take the least fixpoint is Scott-continuous
follows straightforwardly from our previous observations and in particular
Equation~\eqref{eq:scott}.

\begin{example}
        Let us revisit Example~\ref{ex:stop}, where we analysed
        the non-terminating program,
\begin{lstlisting}[mathescape]
x := 0 ; while tt $\{$ x++ ; wait 1 $\}$
\end{lstlisting}
        from an operational point of view. Specifically we saw first-hand the rôle
        that the rules $\textbf{(diff-stop$^\to$)}$ and
        $\textbf{(seq-stop$^\to$)}$ take in giving rise to time-based
        terminations. Let us now illustrate this aspect denotationally, \ie\
        how the combined monad $\Giry \otimes E$ handles such
        terminations via its Kleisli composition. First, 
        for simplicity we denote the loop simply by
        \lstinline|p| and the store 
$\sigma : \{$ \lstinline|x| $\} \to \Reals$
that is defined as $\sigma($\lstinline|x|$) = v$ by \lstinline|x| $\mapsto v$.
Also as in Example~\ref{ex:stop} we consider the time instant $1 +
\sfrac{1}{2}$.  Then,
\begin{align*}
        & \sem{\text{\lstinline|x := 0 ; p|}}(\sigma,1 + \sfrac{1}{2})
        \\
        & = \{ \text{ Semantics definition + Monad laws } \}
        \\
        & \sem{\text{\lstinline|p|}}(\text{\lstinline{x}} \mapsto 0,1 + \sfrac{1}{2})
        \\
        & = \{ \text{ Fixpoint equation } \}
        \\
        & \sem{\text{\lstinline|p|}}^\star
        \comp \sem{\text{\lstinline|x++ ; wait 1|}}
        (\text{\lstinline{x}} \mapsto 0,1 + \sfrac{1}{2})
        \\
        & = \{ \text{ Semantics definition + Monad laws } \}
        \\
        & \sem{\text{\lstinline|p|}}(\text{\lstinline{x}} \mapsto 1,\sfrac{1}{2})
        \\
        & = \{ \text{ Fixpoint equation } \}
        \\
        & \sem{\text{\lstinline|p|}}^\star
        \comp \sem{\text{\lstinline|x++ ; wait 1|}}
        (\text{\lstinline{x}} \mapsto 1,\sfrac{1}{2})
        \\
        & = \{ \text{ Semantics definition } \}
        \\
        & \sem{\text{\lstinline|p|}}^\star
        (1 \cdot (\text{\lstinline{x}} \mapsto 2) )
        \\
        & = \{ \text{ Kleisli composition definition (Equation~\eqref{eq:combined}) } \}
        \\
        & 1 \cdot (\text{\lstinline{x}} \mapsto 2)
\end{align*}
\end{example}

The section's remainder is devoted to proving adequacy of our denotational
semantics w.r.t. the operational counterpart that was described in
Section~\ref{sec:lng}.  In order to achieve this -- and following the same
steps as~\cite{kozen79} -- we will recur to an auxiliary semantics, which
reframes our operational semantics as a \emph{measurable} map. We will see that
such is necessary in order to sensibly extend the input-output relation induced
by the operational semantics to a probabilistic setting -- and thus
subsequently connect the latter to the denotational semantics, as intended.
Let us thus proceed by presenting this auxiliary semantics.

We will need some preliminaries. Recall that any object $E$ in a category
$\catC$ with binary coproducts induces a monad $E + (-)$. We take the
particular case in which $\catC = \Meas$ and $E = 1$. We then equip the Kleisli
morphisms of this monad with the partial order that is induced from the notion
of a flat domain~\cite{GHK+80}. It is easy to see that this order is
$\omega$-complete by an appeal to the following theorem.

\begin{theorem}
Consider an increasing sequence of measurable maps $(f_i)_{i \in \omega} : X \to
1 + Y$. Their supremum w.r.t. the order of flat domains is also measurable.
\end{theorem}

\begin{proof}
        Recall that the supremum $f : X \to 1 + Y$ of $(f_i)_{i \in \omega}$ is
        defined by,
        \[
                f(x) = \begin{cases}
                        \inr(y) & \text{ if }\, \exists i \in \omega. \, 
                        f_i(x) = \inr(y) \text{ for some } y \in Y
                        \\
                        \inl(\ast) & \text{ otherwise }
                \end{cases}
        \]
        (see details for example in~\cite{GHK+80}). Now, we need to show that
        $f$ is measurable, or in other words that both pre-images
        $f^{-1}(\{\inl(\ast)\})$ and $f^{-1}( \inr[U])$ are measurable (where
        $U$ is any measurable subset of $Y$). For the first case we reason,
        \begin{align*}
                f^{-1}(\{ \inl(\ast) \}) & = 
                \{ x \in X \mid f(x) = \inl(\ast) \}
                \\
                & = \{ x \in X \mid \forall i \in \omega. \, f_i(x) = \inl(\ast) \}
                \\
                & = \bigcap_{i = 0}^{\infty}\, f_i^{-1}(\{ \inl (\ast) \})
        \end{align*}
        and this intersection must be measurable because $\sigma$-algebras are
        closed under countable intersections. The second case follows from 
        an analogous reasoning and the fact that $\sigma$-algebras are closed
        under countable unions.
\end{proof}
Observe then that for any increasing sequence of measurable
maps $(f_i)_{i \in \omega} : X \to 1 + Y$ we have,
\[
        (\sup_{i \in \omega} f_i)^\star = \sup_{i \in \omega} f^\star_i
\]
thanks to Scott-continuity of co-pairing on its second argument.  Denoting this
monad by $(-)_\bot$, observe that its Kleisli category $\Meas_{(-)_\bot}$ is
isomorphic to $\PMeas$, \ie\ that of measurable spaces and partial measurable
maps. $\PMeas$ has binary coproducts by general categorical
results~\cite{moggi:1989}. We then take as the interpretation domain of our
auxiliary semantics the monad $E + (-)$ in $\PMeas$ where $E = \Reals^n$. In
order to interpret while-loops via this monad, note that it inherits the
$\omega$-complete order of $(-)_\bot$ and furthermore,
\begin{align}
        \label{eq:scott2}
        (\sup_{i \in \omega} f_i)^{\star} 
        = \sup_{i \in \omega} f_i^{\star}
\end{align}
The operational semantics in functional form is now presented in
Figure~\ref{denot_sem1}. Observe that, contrary to the denotational semantics,
it involves an entropy source. The measurability of each interpretation clause
is once again straightforward to verify, and similarly for the fact that the
map from which we take the least fixpoint is Scott-continuous, thanks to
Equation~\eqref{eq:scott2}. The symbol $\ast$ represents undefinedness. 

Finally the following theorem establishes the aforementioned connection
between the operational semantics in Figure~\ref{fig:big-step2} and the
functional semantics that we have just presented.
\begin{figure*}
\begin{minipage}{1\textwidth}
\begin{align*}
        \asem{  \text{\lstinline[mathescape]|diff(e$_1$,...,e$_n$,e)|}  } & =
        (\sigma,t,s) \mapsto
        \begin{cases}
                \phi(\sigma,t)
                & \text{ if } \sem{\text{\lstinline{e}}}(\sigma) > t
                \\
                (\phi(\sigma,\sem{\text{\lstinline{e}}}(\sigma)),  
                t - \sem{\text{\lstinline{e}}}(\sigma), s)
                & \text{ if } 0 \leq \sem{\text{\lstinline{e}}}(\sigma) \leq t
                \\
                \ast
                & \text{ otherwise }
        \end{cases}
        \\
        \asem{ \text{\lstinline{x := e}} } & = (\sigma,t,s) \mapsto 
        \begin{cases}
                (\sigma[\text{\lstinline{x}} 
                \mapsto \sem{\text{\lstinline{e}}}(\sigma)],t,s) & \text{ if }
                \sem{\text{\lstinline{e}}}(\sigma) \text{ is well-defined }
                \\
                \ast & \text{ otherwise }
        \end{cases}
        \\
        \asem{ \text{\lstinline[mathescape]|x$_i$ := unif(0,1)|} } & = (\sigma,t,(h:s))
        \mapsto
        (\sigma[\text{\lstinline[mathescape]|x$_i$|} 
        \mapsto h],t,s) 
        \\
        \asem{ \text{\lstinline{if b then p else q}} }  & = (\sigma,t,s) \mapsto
        \begin{cases}
                \asem{\text{\lstinline|p|}}\ (\sigma,t,s)
                &
                \text{ if } \sem{\text{\lstinline|b|}}(\sigma) = \mathtt{tt}
                \\
                \asem{\text{\lstinline|q|}}\ (\sigma,t,s)
                &
                \text{ if } \sem{\text{\lstinline|b|}}(\sigma) = \mathtt{ff}
                \\
                \ast & \text{ otherwise}
        \end{cases}
        \\
        \asem{\text{\lstinline{p ; q}}} & = 
        \asem{ \text{\lstinline{q}} }^{\star}
                \comp
                \asem{ \text{\lstinline{p}} }
        \\
        \asem{ \text{\lstinline{while b do p}} } & = 
        \mathrm{lfp} \left (
                k \mapsto 
                (\sigma,t,s) \mapsto
                \begin{cases}
                        k^\star \comp \asem{\text{\lstinline|p|}}\ (\sigma,t,s)
                        & \text{ if } \sem{\text{\lstinline|b|}}(\sigma) =
                        \text{\lstinline|tt|}
                        \\
                        (\sigma,t,s)
                        & \text{ if } \sem{\text{\lstinline|b|}}(\sigma) =
                        \text{\lstinline|ff|}
                        \\
                        \ast
                        & \text{ otherwise } 
                \end{cases}
               \right )
\end{align*}
\caption{Functional version of the big-step semantics in Figure~\ref{fig:big-step2}.}
\label{denot_sem1}
\end{minipage}
\end{figure*}
\begin{theorem}
        Consider a program \textnormal{\lstinline{p}}, an environment $\sigma$, a time
        instant $t$, and an entropy source $s$. Then the following implications
        hold:
        \begin{align*}
          & \textnormal{\lstinline{p}} \sep \sigma \sep  t \sep  s \bsto 
          \skp \sep \sigma' \sep t' \sep s' 
          & \Rightarrow \hspace{0.6cm} 
          & \asem{\textnormal{\lstinline{p}}}(\sigma,t,s) = (\sigma',t',s')
          \\
          & \textnormal{\lstinline{p}} \sep \sigma \sep  t \sep  s \bsto  \sigma' 
          & \Rightarrow \hspace{0.6cm} 
          & \asem{\textnormal{\lstinline{p}}}(\sigma,t,s) = \sigma'
          \\
          & \textnormal{\lstinline{p}} \sep \sigma \sep  t \sep  s \bsto \err 
          & \Rightarrow \hspace{0.6cm} 
          & \asem{\textnormal{\lstinline{p}}}(\sigma,t,s) = \ast
        \end{align*}
        Moreover the following implications also hold:
        \begin{align*}
          & \asem{\textnormal{\lstinline{p}}}\ (\sigma,t,s) = (\sigma',t',s')
          & \Rightarrow \hspace{0.40cm}
          &\textnormal{\lstinline{p}} \sep \sigma \sep  t \sep  s \bsto 
          \skp \sep \sigma' \sep t' \sep s' 
          \\
          & \asem{\textnormal{\lstinline{p}}}\ (\sigma,t,s) = \sigma'
          & \Rightarrow \hspace{0.4cm}
          & \textnormal{\lstinline{p}} \sep \sigma \sep  t \sep  s \bsto  \sigma' 
        \end{align*}
\end{theorem}

\begin{proof}
The proof is laborious but straightforward. The first three implications follow
by induction over the big-step derivations trees, a close inspection of Kleisli
composition, and the fixpoint equation concerning while-loops. The last two
implications follow by structural induction on programs, again a close
inspection of Kleisli composition, and the proof that for all $i \in \omega$
the following implications hold,
\begin{align*}
          & f_i (\sigma,t,s) = (\sigma',t',s')
          & \Rightarrow \hspace{0.40cm}
          &\text{\lstinline|while b do p|} \sep \sigma \sep  t \sep  s \bsto 
          \skp \sep \sigma' \sep t' \sep s' 
          \\
          & f_i (\sigma,t,s) = \sigma'
          & \Rightarrow \hspace{0.4cm}
          & \text{\lstinline|while b do p|} \sep \sigma \sep  t \sep  s \bsto  \sigma' 
        \end{align*}
        where the maps $(f_i)_{i \in \omega} : X \pmap E + X$ are the
        components of the supremum involved in Kleene's least fixpoint
        construction. The proof for these two last implications is obtained
        via induction over the natural numbers.
\end{proof}

Note that the previous theorem excludes the implication,
\[
        \asem{\text{\lstinline|p|}}\ (\sigma,t,s) = \ast \qquad \Rightarrow \qquad
        \text{\lstinline|p|} \sep \sigma \sep t \sep s \bsto \err
\]
This is simply because the equation $\asem{\text{\lstinline|p|}}\ (\sigma,t,s)
= \ast$ may arise from divergence (and not necessarily from an evaluation
error) which the operational semantics cannot track. On the other hand, the
theorem entails that if $\asem{\text{\lstinline|p|}}\ (\sigma,t,s) = \ast$ and
the operational semantics evaluates the tuple $\text{\lstinline|p|} \sep \sigma
\sep t \sep s$ to a value then this value is necessarily $\err$.

We are now ready to extend the auxiliary semantics $\asem{-}$ to a
probabilistic setting. Such hinges on the fact that the pushforward measure
construction $(-)_\star$ is functorial on partial measurable maps.  This yields
the composite functor,
\[
        \xymatrix{
                \PMeas_E \ar[r]^{(-)^\star} &
                \PMeas \ar[r]^{(-)_\star} &
                \Meas
        }
\]
where $\PMeas_E$ is the Kleisli category of the monad  $E + (-)$ in $\PMeas$
and $(-)^\star$ is the respective Kleisli extension. More concretely we obtain
the inference rule,
\[
        \infer{
                (f^\star)_\star : \Giry(E + X) \to \Giry(E + Y)
        }{
                f : X \xrightharpoonup{\hspace{0.1cm}} E + Y
        }
\]
which, intuitively, entails that for any $\asem{\text{\lstinline|p|}}$ the function
$(\asem{\text{\lstinline|p|}}^\star)_\star$ moves probability masses of the
measure given as input according to the operational semantics. This is
analogous to what happens with the denotational semantics; and our adequacy
theorem will render such analogy precise.

We are ready to formulate our adequacy theorem. In order to keep notation
simple we will abbreviate the space $\Reals^n \times \Rz$ to $X$ and, as
before, the space $\Reals^n$ to $E$. Note that any measure $\mu \in
\Giry(E + X)$ can be decomposed into $\mu_{\mid E} \in \Giry(E)$ and $\mu_{\mid
X} \in \Giry(X)$, where $\mu_{\mid E}(U) = \mu(U + \emptyset)$ for all
measurable subsets $U \subseteq E$ and analogously for $\mu_{\mid X}$. These
restriction operations are linear, commute with suprema, and furthermore
$\mu = \inl_\star (\mu_{\mid E}) + \inr_\star (\mu_{\mid X})$.  We
will often abuse notation and elide the left and right injections in the
previous measure decomposition.

\begin{theorem}[Adequacy]
        Consider a program \lstinline|p| and let $\lambda$ be the uniform
        distribution on $[0,1]$. For any measure $\mu \in \Giry(E + X)$, the
        following equation holds.
        \begin{align*}
        & \, (\asem{\text{\lstinline|p|}}^\star)_\star \Big ( \mu_{\mid E}  +
        \mu_{\mid X} \otimes \lambda^{\otimes \omega}  \Big )
        =
        \sem{\text{\lstinline|p|}}^\star (\mu)_{\mid E}
        + 
        \sem{\text{\lstinline|p|}}^\star (\mu)_{\mid X} \otimes
        \lambda^{\otimes \omega}
        \end{align*}
        In particular the equation below holds for all measurable subsets $U
        \subseteq E$ and $V \subseteq X$.
        \begin{align*}
        & \, (\asem{\text{\lstinline|p|}}^\star)_\star
        \Big ( \mu_{\mid E}  + \mu_{\mid X} \otimes
        \lambda^{\otimes \omega} \Big )(U + V \times [0,1]^\omega)
        = \sem{\text{\lstinline|p|}}^\star (\mu) (U + V)
        \end{align*}
\end{theorem}

\begin{proof}
       The proof is obtained via structural induction. The base cases follow
       straightforwardly although laborious. The case of sequential composition
       also follows straightforwardly, thanks to the functorial laws and the
       Kleisli extension laws. 
        We then focus on the case of conditionals. We will need to decompose
        the measure $\mu$ in the cases that it satisfies and does not
        satisfy \lstinline|b|: \ie\ we will denote $\mu(\emptyset +
        \sem{\text{\lstinline|b|}} \cap -)$ by $\nu$ and $\mu(\emptyset +
        \overline{\sem{\text{\lstinline|b|}}} \cap -)$ by $\rho$, where 
        slightly overloading notation we set $\sem{\text{\lstinline|b|}} = \{
                (\sigma,t) \in X \mid \sem{\text{\lstinline|b|}}(\sigma) =
                \mathtt{tt} \}$. Then,
        \begin{align*}
        & \, (\asem{\text{\lstinline|if b then p else q|}}^\star)_\star
                \Big (
                        \mu_{\mid E} + 
                                \mu_{\mid X} \otimes \lambda^{\otimes \omega}
                \Big )
        \\
        & = \big \{ \mu_{\mid X} = (\nu + \rho)_{\mid X} \big \}
        \\
        & \, (\asem{\text{\lstinline|if b then p else q|}}^\star)_\star
                \Big (
                        \mu_{\mid E} + 
                                (\nu + \rho)_{\mid X} \otimes \lambda^{\otimes \omega}
                \Big )
        \\
        & = \big \{ \text{ Addition of measures distributes over tensor } \big \}
        \\
        & \,  (\asem{\text{\lstinline|if b then p else q|}}^\star)_\star
                \Big (
                         \mu_{\mid E} + 
                                \nu_{\mid X} \otimes \lambda^{\otimes \omega}
                                + 
                                \rho_{\mid X} \otimes \lambda^{\otimes \omega}
                \Big )
        \\
        & = \big \{ \text{ Linearity + Semantics definition } \big \}
        \\
        & \,  \mu_{\mid E}
        + (\asem{\text{\lstinline|p|}}^\star)_\star
                \Big (
                        \nu_{\mid X} \otimes \lambda^{\otimes \omega}
                \Big )
          + (\asem{\text{\lstinline|q|}}^\star)_\star
                \Big (
                        \rho_{\mid X} \otimes \lambda^{\otimes \omega}
                \Big )
        \\
        & = \big \{ \text{ Induction hypothesis }\big \}
        \\
        & \,  \mu_{\mid E}
        +
        \sem{\text{\lstinline|p|}}^\star(\nu)_{\mid E} 
        +  
        \sem{\text{\lstinline|p|}}^\star(\nu)_{\mid X} 
                \otimes \lambda^{\otimes \omega}
        +
        \sem{\text{\lstinline|q|}}^\star(\rho)_{\mid E} 
        + 
         \sem{\text{\lstinline|q|}}^\star(\rho)_{\mid X} 
                \otimes \lambda^{\otimes \omega}
        \\
        & = \big \{ \text{ Addition of measures distributes over tensor } \big \}
        \\
        & \,  \mu_{\mid E}
        +
                (\sem{\text{\lstinline|p|}}^\star(\nu)
                +
                \sem{\text{\lstinline|q|}}^\star(\rho))_{\mid E}
        +  
        (\sem{\text{\lstinline|p|}}^\star(\nu)
               + 
              \sem{\text{\lstinline|q|}}^\star(\rho)
        )_{\mid X}
                \otimes \lambda^{\otimes \omega}
        \\
        & = \big \{ \text{ Semantics definition } \big \}
        \\
        & \,
        \sem{\text{\lstinline|if b then p else q|}}^\star(\mu)_{\mid E} 
        +  
        (\sem{\text{\lstinline|p|}}^\star(\nu)
               + 
              \sem{\text{\lstinline|q|}}^\star(\rho))_{\mid X}
                \otimes \lambda^{\otimes \omega}
        \\
        & = \big \{ \text{ Semantics definition } \big \}
        \\
        & \, 
        \sem{\text{\lstinline|if b then p else q|}}^\star(\mu)_{\mid E} 
        +  
        \sem{\text{\lstinline|if b then p else q|}}^\star(\mu)_{\mid X}
                \otimes \lambda^{\otimes \omega}
        \end{align*}
        Lastly we focus on the case of while-loops. Thus let $(f_i)_{i \in \omega}$ be
        the family of maps involved in Kleene's fixpoint construction w.r.t.
        the semantics $\asem{-}$ and analogously for $(g_i)_{i \in \omega}$ and
        the semantics $\sem{-}$. We need to show that for all $i \in \omega$,
        \begin{align}
                \label{eq:lemma}
        & \, (f_i^\star)_\star \Big ( \mu_{\mid E}  +
        \mu_{\mid X} \otimes \lambda^{\otimes \omega}  \Big )
        =
         g_i^\star (\mu)_{\mid E}
        + 
        g_i^\star (\mu)_{\mid X} \otimes
        \lambda^{\otimes \omega}
        \end{align}
        This is obtained via induction over the natural numbers. Note that the
        reasoning is similar to the case concerning conditionals so
        we omit this step. Then,
        \begin{align*}
               & \, (\asem{\text{\lstinline|while b do p|}}^\star)_\star 
               (\mu_{\mid E} + \mu_{\mid X} \otimes \lambda^{\otimes \omega})
               \\
               & = \big \{ \text{ Semantics definition } \big \}
               \\
               & \, ((\sup_{i \in \omega} f_i)^\star)_\star 
               (\mu_{\mid E} + \mu_{\mid X} \otimes \lambda^{\otimes \omega})
               \\
               & = \big \{ \text{ Equation~\eqref{eq:scott2} } \big \}
               \\
               & \, ((\sup_{i \in \omega} f_i^\star))_\star 
               (\mu_{\mid E} + \mu_{\mid X} \otimes \lambda^{\otimes \omega})
               \\
               & = \big \{ \text{ Monotone convergence theorem } \big \}
               \\
               & \, (\sup_{i \in \omega} (f_i^\star)_\star) 
               (\mu_{\mid E} + \mu_{\mid X} \otimes \lambda^{\otimes \omega})
               \\
               & = \big \{ \text{ Equation~\eqref{eq:lemma}  } \big \}
               \\ 
               & \,
               \sup_{i \in \omega} g_i^\star (\mu)_{\mid E} 
               +
               g_i^\star(\mu)_{\mid X} \otimes
               \lambda^{\otimes \omega}
               \\
               & = \big \{ \text{Addition commutes with sup. +  
               Tensor distribute over sup. } \big \} 
               \\
               & \,
               \sup_{i \in \omega} g_i^\star (\mu)_{\mid E} 
               +
               \Big ( \sup_{i \in \omega} g_i^\star(\mu)_{\mid X} \Big ) \otimes
               \lambda^{\otimes \omega}
               \\
               & = \big \{ \text{ Equation~\eqref{eq:scott} + Semantics definition } \big \}
               \\
               & \,
               \sem{\text{\lstinline|while b do p|}}^\star(\mu)_{\mid E}
               + 
               \sem{\text{\lstinline|while b do p|}}^\star(\mu)_{\mid X} \otimes
               \lambda^{\otimes \omega}
        \end{align*}
\end{proof}

\section{Conclusions and future work}
\label{sec:conc}

This paper provides a basis towards a programming framework of stochastic
hybrid systems. It is rooted not only on an operational semantics, which we
used in the implementation of an interpreter, but also on a compositional,
measure-theoretic counterpart, with which one can formally reason about program
equivalence and approximating behaviour, among other things. These
contributions open up several interesting research lines. We briefly detail
next the ones we are currently exploring. 

First, the fact that we committed ourselves to a denotational semantics based
on monads will now allow us to capitalise on the more general, extensive theory
of monad-based program semantics. This includes for example the extension of our
semantics with additional computational effects~\cite{moggi:1989}, with
higher-order features and different evaluation
mechanisms~\cite{levy22,kavvos25}, and corresponding logics as well as
predicate transformer perspectives~\cite{goncharov13,hasuo16}. We are already
working on the last two topics~\cite{goncharov13,hasuo16}, for not only they
potentially offer a complementary tool in our framework of stochastic hybrid
programs they would also facilitate a more natural connection between our work
and previous results on deductive verification of stochastic hybrid
systems~\cite{peng15,platzer11}.

Second, we would like to add new probabilistic constructs to our language, most
notably conditioning~\cite{lee19,dahlqvist19} and stochastic differential
equations which have multiple uses in hybrid systems
theory~\cite{platzer11,peng15}.

Third, recall from Section~\ref{sec:bck} and Section~\ref{sec:denot} that in
our semantics program denotations are \emph{contractive} operators $\Giry(E +
X) \to \Giry(E + X)$. Thus following the observations
in~\cite{dahlqvist23,dahlqvist23b}, such forms the basis of a corresponding
theory of \emph{metric} program equivalence.  Concretely, instead of
comparing two program denotations in terms of \emph{equality} we are able
to systematically compare them in terms of distances. In the setting of
stochastic hybrid programming this is a much more desirable approach, as in
practice it is unrealistic to expect that two programs match their outputs with
exact precision. Not only this, reasoning about program distances is extremely
important for computationally simulating such programs, since we are limited to
certain finite precision aspects and thus can frequently only approximate
idealised behaviours. 

The works~\cite{dahlqvist23,dahlqvist23b} can inclusively be used as basis for
a deductive metric equational system w.r.t. our programming language. Very briefly,
the corresponding metric $d(-,=)$ would be induced by the operator norm in
conjunction with the total variation norm (both are detailed in
Section~\ref{sec:bck}). Then~\cite{dahlqvist23,dahlqvist23b} would lead us on the
analysis of how the different program constructs interact with this metric. For
example it is immediate from the \emph{op. cit.} that for any
programs, \lstinline|p|, \lstinline|p'| and \lstinline|q|, \lstinline|q'|, the
following rule holds:
\[
        \infer{
                d(\sem{\text{\lstinline|p ; q|}},
                \sem{\text{\lstinline|p' ; q'|}}) 
                \leq \epsilon_1 + \epsilon_2
        }
        {
                d(\sem{\text{\lstinline|p|}},\sem{\text{\lstinline|p'|}}) 
                \leq \epsilon_1
                \qquad
                d(\sem{\text{\lstinline|q|}},\sem{\text{\lstinline|q'|}}) 
                \leq \epsilon_2
        }
\]
We leave a full acount of such a metric equational system to future work.

\begin{acks}
        This work was financed by National Funds through FCT/MCTES -- \emph{Fundação para a
                Ciência e a Tecnologia, I.P.} (Portuguese Foundation for Science and
        Technology) within project IBEX, with reference
        10.54499/PTDC/CCI-COM/4280/2021.
        It was further supported by national funds through FCT/MCTES 
        within the CISTER Research Unit (UIDP/UIDB/04234/2020) and under the project
Intelligent Systems Associate Laboratory -- LASI (LA/P/0104/2020).
        Finally we are also thankful for the reviewer's helpful feedback.
\end{acks}

\bibliographystyle{ACM-Reference-Format}
\bibliography{biblio}

\end{document}

%% file: macros.tex
\makeatletter
\DeclareRobustCommand{\iscircle}{\mathord{\mathpalette\is@circle\relax}}
\newcommand\is@circle[2]{%
  \begingroup
  \sbox\z@{\raisebox{\depth}{$\m@th#1\bigcirc$}}%
  \sbox\tw@{$#1\square$}%
  \resizebox{!}{\ht\tw@}{\usebox{\z@}}%
  \endgroup
}
\makeatother

\newcommand{\catfont}[1]{\mathsf{#1}}

\newcommand{\catC}{\catfont{C}}

\newcommand{\CTop}{\catfont{Top}}

\newcommand{\Subs}[2]{\catfont{Sub}_{}}

\newcommand{\Meas}{\catfont{Meas}}
\newcommand{\PMeas}{\catfont{PMeas}}

\newcommand{\sfunfont}[1]{\mathrm{#1}}

\newcommand{\Giry}{\sfunfont{G}}
\newcommand{\Measu}{\sfunfont{M}}

\newcommand{\klcomp}{\star}

\newcommand\adjunct[2]{\xymatrix@=8ex{\ar@{}[r]|{\top}\ar@<1mm>@/^2mm/[r]^{{#2}}
& \ar@<1mm>@/^2mm/[l]^{{#1}}}}
\newcommand\adjunctop[2]{\xymatrix@=8ex{\ar@{}[r]|{\bot}\ar@<1mm>@/^2mm/[r]^{{#2}}
& \ar@<1mm>@/^2mm/[l]^{{#1}}}}

\newcommand\retract[2]{\xymatrix@=8ex{\ar@{}[r]|{}\ar@<1mm>@/^2mm/@{^{(}->}[r]^{{#2}}
& \ar@<1mm>@/^2mm/@{->>}[l]^{{#1}}}}

\newcommand{\inl}{\mathrm{inl}}
\newcommand{\inr}{\mathrm{inr}}

\newcommand{\eval}{\mathrm{eval}}

\newcommand{\comp}{\cdot}

\newcommand{\dist}{\mathrm{dist}}

\newcommand{\sem}[1]{\left \llbracket #1 \right \rrbracket}
\newcommand{\asem}[1]{ \llparenthesis #1 \rrparenthesis}

\newcommand{\Reals}{\mathbb{R}}

\newcommand{\Rz}{\Reals_{\geq 0}}

\newcommand{\ie}{\emph{i.e.}}
\newcommand{\eg}{\emph{e.g.}}

\newcommand{\wrap}[1]{\begin{tabular}{@{}c@{}}#1\end{tabular}}